\documentclass[conference]{IEEEtran}
\usepackage{amsfonts}
\usepackage{amsmath}
\usepackage{verbatim}
\usepackage{setspace}
\usepackage[caption=false,font=footnotesize]{subfig}

\newtheorem{theorem}{Theorem}
\newtheorem{lemma}{Lemma}
\newtheorem{proposition}{Proposition}
\newtheorem{corollary}{Corollary}

\usepackage{cite}

\ifCLASSINFOpdf
   \usepackage[pdftex]{graphicx}
\else
\fi
\newcommand{\los}{{\mathrm{L}}}
\newcommand{\nlos}{{\mathrm{N}}}

\def\bb0{{\mathbb{0}}}

\def\ba{{\mathbf{a}}}
\def\bb{{\mathbf{b}}}

\def\bff{{\mathbf{f}}}

\def\bn{{\mathbf{n}}}

\def\br{{\mathbf{r}}}

\def\bw{{\mathbf{w}}}

\def\b0{{\mathbf{0}}}

\def\bH{{\mathbf{H}}}

\def\bbE{{\mathbb{E}}}

\def\sf0{{\mathsf{0}}}

\usepackage{algorithm}
\usepackage{algpseudocode}
\usepackage{epstopdf}
\usepackage{balance}
\usepackage{bbm}
\IEEEoverridecommandlockouts

\begin{document}
%
\title{Millimeter Wave Energy Harvesting}



\author{Talha Ahmed Khan, Ahmed Alkhateeb, and Robert W. Heath Jr. \thanks{The authors are with the Wireless Networking and Communications Group at The University of Texas at Austin (Email: \{talhakhan, aalkhateeb, rheath\}@utexas.edu).} \thanks{This work was supported in part by the Army Research Office under grant W911NF-14-1-0460, and gifts from Mitsubishi Electric Research Labs, Cambridge and Nokia.} \thanks{Part of this work was presented at IEEE GLOBECOM Workshops, 2015~\cite{talha2015gc}.}}
\maketitle
\setcounter{page}{1} \thispagestyle{plain}
\begin{abstract}
The millimeter wave (mmWave) band, a prime candidate for 5G cellular networks, seems attractive for wireless energy harvesting since it will feature large antenna arrays and extremely dense base station (BS) deployments. 
The viability of mmWave for energy harvesting though is unclear, due to the differences in propagation characteristics such as extreme sensitivity to building blockages. 
This paper considers a scenario where low-power devices extract energy and/or
information from the mmWave signals. 
Using stochastic geometry, analytical expressions are derived for the energy coverage probability, the average harvested
power, and the overall (energy-and-information) coverage probability at a typical wireless-powered
device in terms of the BS density, the antenna geometry
parameters, and the channel parameters. 
Numerical results reveal several network and device level
design insights. At the BSs, optimizing the antenna geometry parameters such as beamwidth can maximize the network-wide energy coverage for a given user population. 
At the device level, the performance can be substantially improved by optimally splitting the received
signal for energy and information extraction, and by deploying multi-antenna arrays. 
For the latter, an efficient low-power multi-antenna mmWave receiver architecture is proposed for simultaneous energy
and information transfer. 
Overall, simulation results suggest that mmWave energy harvesting generally
outperforms lower frequency solutions.
\end{abstract}

%
\IEEEpeerreviewmaketitle

\section{Introduction}
Millimeter wave (mmWave) communications is a key candidate technology for future 5G cellular networks. This is mainly due to the availability of large spectrum resources at higher frequencies, which leads to much higher data rates. Recent research suggests that mmWave systems will typically feature (i) large-dimensional antenna arrays with directional beamforming at the transmitter/receiver---which is motivated by the small wavelength that allows packing a large number of antenna elements into small form-factors; and (ii) a dense deployment of base stations (BSs) to ensure a comparable coverage to ultra high frequency (UHF) networks\cite{rappaport2014millimeter,bai2014b}. These mmWave design features are also attractive for RF (radio frequency) energy harvesting where a harvesting device may extract energy from the incident RF signals\cite{EnergyHarvestWirelessCommSurvey2015}. 
This could potentially power the massive number of low-power wireless devices in future paradigms such as the Internet of Things\cite{IoT2014}.  
The signal propagation at mmWave frequencies, however, suffers from poor penetration and diffraction characteristics, making it sensitive to blockage by buildings\cite{bai2014b,rangan2014millimeter}. 
It is, therefore, unclear if mmWave cellular networks will be more favorable for RF energy harvesting compared to the conventional (below 6 GHz) frequencies. 
Further, the network level design principles for mmWave energy harvesting systems are not well understood. 
This motivates a network view of energy harvesting in a mmWave cellular network.

\subsection{Contributions}
In this paper, we provide a tractable framework to characterize the performance of wireless energy and information transfer aided by a large-scale mmWave cellular network. 
Our analysis accounts for the key distinguishing features of mmWave systems, namely the sensitivity to blockage and the use of potentially large antenna arrays at the transmitter/receiver. 
We first consider mmWave energy harvesting, where devices only extract energy from the incident mmWave signals. Our analysis models two operating scenarios, one where devices have their beams aligned to that of a mmWave BS, and the other where no such beam alignment is assumed. 
For both operating modes, we derive simple analytical expressions for metrics such as the energy coverage probability and the average harvested power using tools from stochastic geometry. We then extend the analysis to characterize the overall (energy-and-information) coverage probability for the general case where a device extracts both energy and information from the mmWave signals.   

To get design insights, we examine the network level performance trends in terms of key parameters such as the mmWave network density and the antenna geometry parameters for both operating modes of the energy harvesting devices.
Numerical results suggest that narrower antenna beams are preferred when the users are aligned with a BS, whereas wider beams are favorable when no beam alignment is assumed. 
Our findings also suggest that there typically exists an optimum transmit antenna beamwidth that maximizes the network-wide energy coverage for a given user population. This implies that the mmWave BSs will need to adapt the antenna beam patterns depending on the fraction of the users operating in each mode. 

Similar to the BS-related parameters, we also investigate the role of the device-related parameters on the system performance. 
For example, the overall (energy-and-information) coverage probability can be improved by optimizing over a design parameter (power splitting ratio) to optimally portion the received signal between the energy harvesting and the information decoding modules. Another important design feature at the user is the receive antenna array.
Similar to the BSs, the mmWave users can, in principle, benefit from using large antenna arrays. For low-power energy harvesting devices, however, the associated antenna circuity could increase the power consumption, offsetting the potential gains of large antenna arrays.       
To leverage multiple antennas at the receiver without 
resorting to power-hungry circuit components, we propose a simple switch-based receiver architecture for simultaneous energy and information transfer. 
Simulation results reveal that the proposed low-power solution performs reasonably well compared to more advanced but power-hungry receiver architectures.

\subsection{Related Work}
Wireless energy harvesting is becoming increasingly feasible due to the reduction in the power consumption requirements of wireless sensors and the improvements in energy harvesting technologies \cite{talla2015powering,GollakotaRF,RFsurveyLondon,valenta2014}. 
This has also led to considerable research in advancing the theoretical understanding of wireless-powered systems (see \cite{EnergyHarvestWirelessCommSurvey2015,WPCSurvey2015} for a comprehensive overview). For example, wireless energy and information transfer has been studied for different information-theoretic setups such as a broadcast channel\cite{zhang2013mimo}, a fading channel\cite{SWIPTfading}, and an interference channel \cite{SWIPTInter}. Many of these papers highlight the fundamental trade-off between energy and information transfer efficiency and characterize the achievable rate-energy regions for different practical receiver architectures\cite{WPCSurvey2015}.    

Wireless energy and/or information transfer in large-scale networks has also been investigated\cite{flint2015,huang2013cog,ekram2015,kaibin2014cellular,kaibinlarsson2013,krikidis2014swipt}.
In \cite{flint2015}, the performance of ambient RF energy harvesting was characterized using tools from stochastic geometry. Using a repulsive point process to model RF transmitters, it was shown that more repulsion helps improve the performance at an energy harvester for a given transmitter density. 
In \cite{huang2013cog,ekram2015}, cognitive radio networks were considered, and opportunistic wireless energy harvesting was proposed and analyzed. 
In \cite{kaibin2014cellular}, a hybrid cellular network architecture was proposed to enable wireless power transfer for mobiles. In particular, an uplink cellular network was overlaid with power beacons and trade-offs between the transmit power and deployment densities were investigated under an outage constraint on the data links. 
A broadband wireless network with transmit beamforming was considered in \cite{kaibinlarsson2013}, where optimal power control algorithms were devised for improving the throughput and power transfer efficiency.
Simultaneous information and energy transfer in a relay-aided network was considered in \cite{krikidis2014swipt}. Under a random relay selection strategy, the network-level performance was characterized in terms of the relay density and the relay selection area.  

Our work differs from the prior work in that we investigate wireless energy and information transfer in a large-scale \textit{mmWave} cellular network.
Due to the different physical characteristics and design features at mmWave, prior work on energy/information transfer in lower frequency networks does not directly apply to mmWave networks. 
In another line of work, the performance of mmWave cellular networks in terms of signal-to-interference-and-noise ratio (SINR) coverage and rate has also been analyzed using stochastic geometry\cite{bai2015,Singh2015}. None of this work on mmWave networks, however, provides a performance characterization from the perspective of wireless energy and information transfer.

The paper is organized as follows. In Section \ref{secSys}, we introduce the system model. Section \ref{secEH} presents the analytical results for mmWave energy harvesting. The case with simultaneous information and energy transfer is treated in Section \ref{secSWIPT}. We conclude the paper in Section \ref{secConc}.

\section{System Model}\label{secSys}
In this section, we introduce the network and channel models, followed by a description of the antenna model.
The parameters defined in this section are summarized in Table \ref{table1}.
\subsection{Network Model}\label{secNet}
We consider a large-scale cellular network consisting of mmWave BSs and a population of wireless-powered devices (or users) that operate by extracting energy and/or information in the mmWave band. 
The mmWave BSs are located according to a homogeneous Poisson point process (PPP) $\Phi(\lambda)$ of density $\lambda$.
The user population is drawn from another homogeneous PPP $\Phi_u(\lambda_u)$ of density $\lambda_u$, independently of $\Phi$.
In general, mmWave BSs and users may be located outdoors or indoors. Empirical evidence suggests that mmWave signals exhibit high penetration losses for many common building materials\cite{rangan2014millimeter,bai2015}. Assuming the building blockages to be impenetrable, we focus on the case where the BSs and users are located outdoors. 
We say that a BS-user link is line-of-sight (LOS) or non-line-of-sight (NLOS) depending on whether or not it is intersected by a building blockage. Channel measurement campaigns have reported markedly different propagation characteristics for LOS/NLOS links\cite{rappaport2014millimeter,rangan2014millimeter}.
To model blockage due to buildings, we leverage the results in \cite{bai2014} where the buildings are drawn from a boolean stochastic point process. 
We define a line-of-sight (LOS) probability function $p(r)=e^{-\beta r}$ for a link of length $r$, where $\beta$ is a constant that depends on the geometry and density of the building blockage process: a BS-receiver link of length $r$ is declared LOS with a probability $p(r)$, independently of other links. While conducting stochastic geometry analysis, we will apply this result to split the BS PPP into two independent but non-homogeneous PPPs consisting of LOS and NLOS BSs.

We allow the user population to consist of two types of users, namely \emph{connected} and \emph{nonconnected}. A connected user is assumed to be tagged with the BS, either LOS or NLOS, that maximizes the average received power at that user. Moreover, for the connected case, we assume perfect beam alignment between a BS and its tagged user, i.e., the BS and user point their beams  so as to have the maximum directivity gain. Further, we assume that a BS serves only one connected user at a given time. For a nonconnected user, we do not assume any prior beam alignment with a BS, i.e., it is not tagged with any BS.
This allows us to model a wide range of scenarios. For instance, due to limited resources, the mmWave network may (directly) serve only a fraction of the user population as connected users, leaving the rest in the nonconnected mode. Another interpretation could be that due to the challenges associated with channel acquisition, not all the users could be simultaneously served in the connected mode. We let $\epsilon$ be the probability that a randomly selected node is a connected user, independently of other nodes. With this assumption, we can thin the user PPP $\Phi_{u}$ into two independent PPPs $\Phi_{u,\textrm{con}}$ and $\Phi_{u,\textrm{ncon}}$, with respective densities $\epsilon\lambda_u$ and $(1-\epsilon)\lambda_u$. Note that an arbitrary user, either connected or nonconnected, may experience an energy outage if the received power falls short of a required threshold $\psi$. 
This threshold would depend on the power consumption requirements of the receiver. 
To capture the sensitivity requirements of the harvesting circuit, we define $\psi_\textrm{min}$ to be the harvester activation threshold, i.e., the minimum received energy needed to activate the harvesting circuit (the energy outage threshold $\psi$ would typically be greater than $\psi_\textrm{min}$).
We use $\xi$ to denote the rectifier efficiency.
We define $\textrm{P}_\textrm{con}\left(\lambda,\psi_\textrm{con}\right)$ to be the energy coverage probability given an outage threshold $\psi_\textrm{con}$ for a connected user, while $\textrm{P}_\textrm{ncon}\left(\lambda,\psi_\textrm{ncon}\right)$ denotes the same for the nonconnected case.
With these definitions, we can define the overall energy coverage probability $\Lambda(\epsilon,\lambda,\psi_\textrm{con},\psi_\textrm{ncon})$ of the network as
\begin{small}\begin{align}
\Lambda(\epsilon,\lambda,\psi_\textrm{con},\psi_\textrm{ncon})=\epsilon\textrm{P}_\textrm{con}\left(\lambda,\psi_\textrm{con}\right)+
(1-\epsilon)\textrm{P}_\textrm{ncon}\left(\lambda,\psi_\textrm{ncon}\right)
\end{align}\end{small}where the energy coverage probability is a function of several parameters such as the BS density, the channel propagation parameters, as well as the antenna beam patterns at the transmitter/receiver. 
For cleaner exposition, we drop the subscript in  $\psi_\textrm{con}$ or  $\psi_\textrm{ncon}$, using the notation $\Lambda(\epsilon,\lambda,\psi)$ when the context is clear.
In Section \ref{secStoch}, we provide analytical expressions to compute the energy coverage probability in a mmWave network.
\subsection{Channel Model}\label{secCha}
We now describe the channel model for an arbitrary user without losing generality. Empirical evidence suggests that mmWave frequencies exhibit different propagation characteristics for the LOS/NLOS links\cite{rangan2014millimeter}. While the LOS mmWave signals propagate as if in free space, the NLOS mmWave signals typically exhibit a higher path loss exponent (and additional shadowing)\cite{rangan2014millimeter}. We let $\alpha_\los$ and $\alpha_\nlos$ be the path loss exponents for the LOS and NLOS links respectively. We define the distance-dependent path loss for a user located a distance $r_\ell$ from the $\ell$-th BS: $g_\ell(r_\ell)=C_\los r_\ell^{-\alpha_\los}$ when the link is LOS, where the constant $C_\los$ is the path loss intercept; and $g_\ell(r_\ell)=C_\nlos r_\ell^{-\alpha_\nlos}$ for the NLOS case. 
Note that by including blockages in our model (Section \ref{secNet}), we capture the distance-dependent signal attenuation due to buildings. To simplify the analysis, we do not include additional forms of shadowing in our model. 
We further define $h_{\ell}$ to be the small-scale fading coefficient corresponding to a BS $\ell\in\Phi$. Assuming independent Nakagami fading for each link, the small-scale fading power $H_\ell=|h_{\ell}|^2$ can be modeled as a normalized Gamma random variable, i.e., $H_{\ell}\,${\raise.17ex\hbox{$\scriptstyle\mathtt{\sim}$}}$ \,\Gamma\left(N_\los,1/N_\los\right)$ when the link is LOS and $H_{\ell}\,
${\raise.17ex\hbox{$\scriptstyle\mathtt{\sim}$}}$\,\Gamma\left(N_\nlos,1/N_\nlos\right)$ for the NLOS case, where the fading parameters $N_\los$ and $N_\nlos$ are assumed to be integers for simplicity. 
\subsection{Antenna Model}\label{secAnt}
To compensate for higher propagation losses, mmWave BSs will use large directional antennas arrays. We assume that the BSs and users are equipped with $N_{\textrm{t}}$ and $N_{\textrm{r}}$ antenna elements each. To simplify the analysis while capturing the key antenna characteristics, we use the sectored antenna model of Fig. \ref{fig:ant_pat} (except for Section IV), similar to the one considered in \cite{bai2015,hunter2008}. We use $A_{M,m,\theta,\bar{\theta}}(\phi)$ to characterize the antenna beam pattern, where $\phi$ gives the angle from the boresight direction, $M$ denotes the directivity gain and $\theta$ the half power beamwidth for the main lobe, while $m$ and $\bar{\theta}$ give the corresponding parameters for the side lobe. 
With this notation, $A_{M_\textrm{t},m_\textrm{t},\theta_\textrm{t},\bar{\theta}_\mathrm{t}}(\cdot)$ denotes the antenna beam pattern at an arbitrary BS in $\Phi$, and  $A_{M_\textrm{r},m_\textrm{r},\theta_\textrm{r},\bar{\theta}_\mathrm{r}}(\cdot)$ denotes the same for an energy harvesting user in $\Phi_u$. We further define $\delta_\ell=A_{M_\textrm{t},m_\textrm{t},\theta_\textrm{t},\bar{\theta}_\mathrm{t}}(\phi^{\ell}_\textrm{t})
A_{M_\textrm{r},m_\textrm{r},\theta_\textrm{r},\bar{\theta}_\mathrm{r}}(\phi^{\ell}_\textrm{r})$, the total directivity gain for the link between the $\ell$-th BS and the typical user; $\phi^{\ell}_\textrm{t}$ and $\phi^{\ell}_\textrm{r}$ give the angle-of-arrival and angle-of-departure of the signal.
\begin{figure} [t]
	\centerline{
		\includegraphics[width=.5\columnwidth]{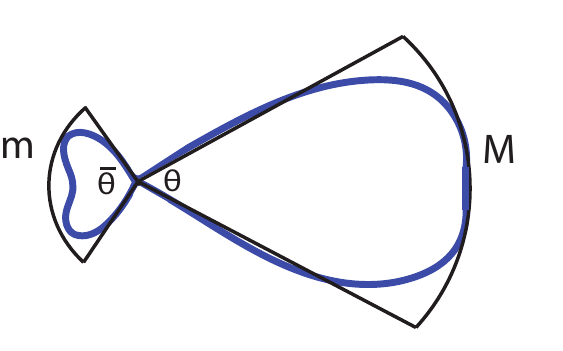}
	}
	\caption{Sectored antenna model. The antenna beam pattern is parameterized by the directivity gains for the main lobe ($M$) and side lobe ($m$), and the half power beamwidths for the main lobe ($\theta$) and side lobe ($\bar{\theta}$).}
	\label{fig:ant_pat}
\end{figure}

Without any further assumptions about the beam alignment between a user and its BS, we model the directivity gain $\delta_\ell$ as a random variable.
We assume the angles $\phi^{\ell}_\textrm{t}$ and $\phi^{\ell}_\textrm{r}$ are uniformly distributed in $[0,2\pi)$. Due to the sectored antenna model, the random variable $\delta_\ell=D_i$ with a probability $p_i$ ($i\in\{1,2,3,4,5\}$), where $D_i\in\{M_\textrm{t}M_\textrm{r},M_\textrm{t}m_\textrm{r},m_\textrm{t}M_\textrm{r},m_\textrm{t}m_\textrm{r},0\}$ with corresponding probabilities $p_i\in\{q_\mathrm{t}{q}_\mathrm{r},q_\mathrm{t}\bar{q}_\mathrm{r},\bar{q}_\mathrm{t}{q}_\mathrm{r},\bar{q}_\mathrm{t}\bar{q}_\mathrm{r},q_o\}$; the constants $q_\mathrm{t}=\frac{\theta_t}{2\pi}$, $\bar{q}_\mathrm{t}=\frac{\bar{\theta}_t}{2\pi}$, $q_\mathrm{r}=\frac{\theta_r}{2\pi}$, $\bar{q}_\mathrm{r}=\frac{\bar{\theta}_r}{2\pi}$, and $q_o=2-q_\mathrm{t}- \bar{q}_\mathrm{t}-q_\mathrm{r}- \bar{q}_\mathrm{r}$. Note that $D_5=0$ models the extreme case where the BS and user beams have no alignment at all.  
Note that for the connected mode, since we assume perfect beam alignment between the typical user and its serving BS (hereby denoted by subscript $0$), the directivity gain $\delta_0=M_\textrm{t}M_\textrm{r}$ due to the sectored antenna model.

\begin{table}
	\caption{Model Parameters}	
	\centering
	\begin{tabular}{| p{0.7in} || p{2.0in}|}
		\hline
		\textbf{Notation} & \textbf{Description} \\ \hline
		$N_\textrm{t}$, $N_\textrm{r}$ & Antenna array size at the transmitter (t) and receiver (r)  \\ \hline
		$M_\textrm{t}$, $M_\textrm{r}$ \newline $m_\textrm{t}$, $m_\textrm{r}$ & Main lobe directivity gain \newline Side lobe directivity gain\\ \hline
		$\theta_\textrm{t}$, $\theta_\textrm{r}$ \newline $\bar{\theta}_\textrm{t}$, $\bar{\theta}_\textrm{r}$ & Main lobe half power beamwidth \newline Side lobe half power beamwidth\\ \hline
		$\Phi(\lambda)$ & BS PPP with density $\lambda$ \\ \hline
		$\Phi_{u}(\lambda_u)$ & User PPP with density $\lambda_u$ \\ \hline
    	$\epsilon$ & Fraction of connected users \\ \hline
    	$\psi$ & Energy outage threshold \\ \hline
    	$\psi_\textrm{min}$ & Harvester activation threshold\\ \hline
 	    $\xi$ & Rectifier efficiency\\ \hline
    	$\Lambda(\epsilon,\lambda,\psi)$ & Energy coverage probability \\ \hline
    	$p(r)$ & LOS probability function \\ \hline
    	$\beta$ & Building blockage parameter \\ \hline
        $\alpha_\los$, $\alpha_\nlos$ & LOS/NLOS path loss exponents \\ \hline
        $C_\los$, $C_\nlos$ & LOS/NLOS path loss intercepts \\ \hline
        $N_\los$, $N_\nlos$ & LOS/NLOS fading parameters \\ \hline	
    	$P_\textrm{t}$ & Transmit power of BSs in $\Phi$ \\ \hline
	\end{tabular}
	\label{table1}
\end{table}
\section{MmWave with Energy Harvesting}\label{secEH}
In this section, we assume that each user is equipped with an energy harvesting circuit, and attempts to extract energy from the incident mmWave signals. No decoding of information is considered in this section. The case with simultaneous information and power transfer is treated in Section \ref{secSWIPT}. 
We first provide analytical expressions to evaluate the energy coverage probabilities for both connected and nonconnected users. We then validate the analytical model, and conclude the section by providing network level design insights.

\subsection{Stochastic Geometry Analysis}
\label{secStoch}
We first provide some lemmas before stating the main analytical results for this section.
\begin{lemma}[{Modified from\cite[Theorem 8]{bai2014}}]\normalfont
The probability density function (PDF) of the distance from an energy harvesting user to its nearest LOS BS, given that the user observes at least one LOS BS, is given by ${\tau}_\los\left(x\right)=2\pi\lambda {B_\los}^{-1}x p(x) e^{-2\pi\lambda\int_{0}^{x}vp(v)\text{d}v}
$, where $x>0$ and ${B_\los}=1-e^{-2\pi\lambda\int_{0}^{\infty}vp(v)\text{d}v}$ is the probability that the receiver observes at least one LOS BS. Similarly, the distance distribution of the link between the user and its nearest NLOS BS, given that the user observes at least one NLOS BS, is given by
${\tau}_\nlos\left(x\right)=2\pi\lambda {B_\nlos}^{-1}x (1-p(x)) e^{-2\pi\lambda\int_{0}^{x}v(1-p(v))\text{d}v}$, where $x>0$ and ${B_\nlos}=1-e^{-2\pi\lambda\int_{0}^{\infty}v(1-p(v))\text{d}v}$ is the probability that the user observes at least one NLOS BS.
\end{lemma}
\begin{lemma}[{Modified from \cite[Lemma 2]{bai2015}}]\normalfont
Let $\varrho_\los$ and $\varrho_\nlos$ denote the probability that the energy harvesting user is connected to a LOS and a NLOS BS respectively, then $\varrho_{\mathrm{L}}$ is given by $\varrho_{\mathrm{L}}=B_{\mathrm{L}}\int_{0}^{\infty}e^{-2\pi\lambda\int_{0}^{\rho_{\mathrm{L}}(x)}(1-p(v))v\text{d}v}\tau_{\mathrm{L}}\left(x\right)\text{d}x,
$
where $\rho_{\mathrm{L}}(x)=\left(\frac{C_{\mathrm{N}}}{C_{\mathrm{L}}}\right)^{\frac{1}{\alpha_\nlos}}x^{\frac{\alpha_\los}{\alpha_\nlos}}$ and $\varrho_\nlos=1-\varrho_\los$.
\end{lemma}
\begin{lemma}[{Modified from \cite[Lemma 3]{bai2015}}]\normalfont
Given that the energy harvesting user is connected to a LOS mmWave BS, the PDF of the link distance is given by the expression
$\tilde{\tau}_\los\left(x\right)=\frac{B_\los {\tau}_\los\left(x\right)}{\varrho_\los}e^{-2\pi\lambda\int_{0}^{\rho_\los(x)}(1-p(v))v\text{d}v},$
where $x>0$. Given that the user is connected to a NLOS mmWave BS, the PDF of the link distance is given by
$\tilde{\tau}_\nlos\left(x\right)=\frac{B_\nlos {\tau}_\nlos\left(x\right)}{\varrho_\nlos}e^{-2\pi\lambda\int_{0}^{\rho_\nlos(x)}p(v)v\text{d}v}$
for $x>0$ and $\rho_\nlos(x)=\left(\frac{C_{\mathrm{L}}}{C_{\mathrm{N}}}\right)^{\frac{1}{\alpha_\los}}x^{\frac{\alpha_\nlos}{\alpha_\los}}$.
\end{lemma}

Leveraging Slivnyak's theorem\cite{haenggi2012stochastic}, we conduct the analysis at a typical energy harvesting user located at the origin without losing generality. We let $P_\mathrm{t}$ be the BS transmit power, and $Y=\sum_{\ell\in\Phi(\lambda)}^{}P_\textrm{t}\delta_\ell H_\ell g_\ell(r_\ell)$ be the power received at the user. Recall that $\psi_\textrm{min}$ denotes the harvester activation threshold defined in Section \ref{secNet}. The energy harvested at a typical receiver (in unit time) can be expressed as
\begin{align}\label{eq1}
\gamma=\xi Y\mathbbm{1}_{\{Y>\psi_{\textrm{min}}\}} 
\end{align}
where $\xi\in(0,1]$ is the rectifier efficiency. Note that we have neglected the noise term since it is extremely small relative to the aggregate received signal. 
The remaining parameters follow from Section \ref{secSys}. Recall that given a BS $\ell\in\Phi(\lambda)$, the corresponding fading parameters will be distinct depending on whether the link is LOS or NLOS, which in turn depends on the LOS probability function (Section \ref{secNet}).
Further note that for the connected case, it follows from Section \ref{secAnt} that $\delta_0=M_\textrm{t}M_\textrm{r}$ for the link from the serving BS (denoted by subscript $0$).
\subsubsection*{\textbf{Connected case}}
The following theorem provides an analytical expression for the energy coverage probability $\textrm{P}_\textrm{con}\left(\lambda,\psi\right)=\Pr\{\gamma>\psi\}$ at a connected user, where the random variable $\gamma$ is given in (\ref{eq1}), and $\psi$ is the energy outage threshold. Note that $\textrm{P}_\textrm{con}\left(\lambda,\psi\right)$ can also be interpreted as the complementary cumulative distribution function (CCDF) of the harvested energy.
\begin{theorem}\normalfont
\label{thm1}
In a mmWave network with density $\lambda$, the energy coverage probability $\textrm{P}_\textrm{con}\left(\lambda,\psi\right)$ for the connected case given an energy outage threshold $\psi$, can be evaluated as
\begin{align}
\textrm{P}_\textrm{con}\left(\lambda,\psi\right)&= 
\textrm{P}_\textrm{con,L}\left(\lambda,\hat{\psi}\right)\varrho_\textrm{L}+\textrm{P}_\textrm{con,N}\left(\lambda,\hat{\psi}\right)\varrho_\textrm{N},
\end{align}
where $\hat{\psi}=\max\left(\frac{\psi}{\xi},\psi_\textrm{min}\right)$, $\varrho_\textrm{L}=1-\varrho_\textrm{N}$ is given in Lemma 2, while $\textrm{P}_\textrm{con,L}\left(\cdot\right)$ and $\textrm{P}_\textrm{con,N}\left(\cdot\right)$ are the conditional energy coverage probabilities given the serving BS is LOS or NLOS. These terms can be tightly approximated as
\begin{align}\label{los}
\textrm{P}_\textrm{con,L}\left(\lambda,\psi\right)&\approx 
\sum\limits_{k=0}^{N}\left(-1\right)^{k}\binom{N}{k}~~\times  \nonumber\\
\int\limits_{r_g}^{\infty}\zeta_k^\los\left(r\right)& e^{-{\overset{}{\Upsilon}}_{k,1}^{}\left(\lambda,\psi,r\right) -{\overset{}{\Upsilon}_{k,2}\left(\lambda,\psi,\rho_\los(r)\right)}}\tilde{\tau}_\los\left(r\right)\text{d}r,
\end{align}
where $\zeta_k^\los\left(x\right)=\left(1+\frac{akP_\textrm{t} M_{\mathrm{t}}M_{\mathrm{r}}C_\los}{ \psi{N_\los}x^{\alpha_\los}}\right)^{-N_\los}
$, the approximation constant $a=N(N!)^{-\frac{1}{N}}$ where $N$ denotes the number of terms in the approximation, while $r_g$ defines the minimum link distance and is included to avoid unbounded path loss at the receiver. Similarly,
\begin{align}\label{nlos}
\textrm{P}_\textrm{con,N}\left(\lambda,\psi\right)&\approx 
\sum\limits_{k=0}^{N}\left(-1\right)^{k}\binom{N}{k}~~\times \nonumber\\
\int\limits_{r_g}^{\infty}\zeta_k^\nlos\left(r\right) & e^{-{\overset{}{\Upsilon}}_{k,1}\left(\lambda,\psi,\rho_\nlos(r)\right) -{\overset{}{\Upsilon}}_{k,2}\left(\lambda,\psi,r\right)}\tilde{\tau}_\nlos\left(r\right)\text{d}r,
\end{align}
where $\zeta_k^\nlos\left(x\right)=\left(1+\frac{akP_\textrm{t} M_{\mathrm{t}} M_{\mathrm{r}} C_\nlos}{ \psi{N_\nlos}x^{\alpha_\nlos}}\right)^{-N_\nlos}
$,
\begin{small}\begin{align}\label{fac1}
{\overset{}{\Upsilon}}_{k,1}\left(\lambda,\psi,x\right)=2\pi\lambda\sum_{i=1}^{4}p_i \int\limits_{x}^{\infty}&\left(1-{\left[1+\frac{aP_\textrm{t}kD_iC_\los} {\psi{N_\los}t^{\alpha_\los}}\right]^{-N_\los}}\right)\nonumber\\
&\hspace{0.4in}\times p(t)t\text{d}t,
\end{align}
\end{small}
\begin{small}
\begin{align}\label{fac2}
{\overset{}{\Upsilon}}_{k,2}\left(\lambda,\psi,x\right)=
2\pi\lambda\sum_{i=1}^{4}p_i \int\limits_{x}^{\infty}&\left(1-{\left[1+\frac{aP_\textrm{t}kD_iC_\nlos} {\psi{N_\nlos}t^{\alpha_\nlos}}\right]^{-N_\nlos}}\right)\nonumber\\
&\hspace{0.4in}\times
\left(1-p(t)\right)t\text{d}t,
\end{align}
\end{small}and the distance distributions $\tilde{\tau}_\los(\cdot)$ and $\tilde{\tau}_\nlos(\cdot)$ follow from Lemma 3.
\end{theorem}
\begin{proof}
See Appendix A.
\end{proof}
Recall that $p(t)=e^{-\beta t}$ is the LOS probability function defined in Section \ref{secNet}, and captures the effect of building blockages. In (\ref{los}), the term $\zeta_k^\los\left(\cdot\right)$ models the contribution from the LOS serving link, ${\overset{}{\Upsilon}}_{k,1}\left(\cdot\right)$ accounts for the other LOS links, and ${\overset{}{\Upsilon}}_{k,2}\left(\cdot\right)$ captures the effect of the NLOS links. Note that the $i$th term in (\ref{fac1}), (\ref{fac2}) corresponds to the contributions from the BS-user links having directivity gain $D_i$. Similarly, $\zeta_k^\nlos\left(\cdot\right)$ in (\ref{nlos}) models the case where the serving BS is NLOS. Note that these terms further depend on the channel propagation conditions ($\alpha_\los$, $\alpha_\nlos$, $N_\los$, $N_\nlos$, $C_\los$, $C_\nlos$), the network density $\lambda$ as well as the antenna geometry parameters (via $D_i$, $p_i$), which are summarized in Table \ref{table1}.
Furthermore, the outage threshold $\psi$ will depend on the power requirements at a particular user, and would typically be greater than the sensitivity of the harvesting circuit (i.e., $\psi\geq\psi_\textrm{min}$ such that $\hat{\psi}=\frac{\psi}{\xi}$).
Though the expressions in Theorem 1 can be evaluated using numerical tools, this could be tedious due to the presence of multiple integrals.
To address this, we simplify the analysis by approximating the LOS probability function with a step function, and by further ignoring the small-scale fading. This result in a much simpler expression for the energy coverage probability.   
\begin{proposition}\normalfont
Let $R_B=\left(\frac{-\ln(1-\varrho_\mathrm{L})}{\lambda\pi}\right)^{0.5}$, $\tilde{a}=\lambda\pi R_B^2 e^{-\lambda\pi R_B^2}$, and $W_{ik}=\frac{ak D_i P_\textrm{t}C_\mathrm{L}}{\hat{\psi}}$ for $i\in\{1,2,3,4\}$, $k\in\{0,1,\cdots,N\}$. The energy coverage probability  can be further approximated as 
\begin{align}\label{eq:prop1}
& P_\textrm{con}\left(\lambda,\psi\right)\approx 
\tilde{a}\sum\limits_{k=0}^{N}\left(-1\right)^{k}\binom{N}{k}
\int_{\left(\frac{r_g}{R_B}\right)^2}^{1}
\zeta_k^\los\left({t}^{\frac{1}{2}}R_B\right)
\times\nonumber\\
&
 \prod\limits_{i=1}^{4}e^{-\frac{2\pi\lambda}{\alpha_\textrm{L}}
p_iW_{ik}^{\frac{2}{\alpha_\mathrm{L}}}
\Gamma\left(-\frac{2}{\alpha_\mathrm{L}}; W_{ik}\left(t^{\frac{1}{2}}R_B\right)^{-\alpha_{\mathrm{L}}},W_{ik}R_B^{-\alpha_{\mathrm{L}}} \right) }\text{d}t,
\end{align}
\end{proposition}
where $\Gamma\left(h;u,v\right)=\int_u^v x^{h-1}e^{-x}\text{d}x$ is the generalized incomplete Gamma function.
\begin{proof}
See Appendix B.
\end{proof}
We note that Proposition 1 is relatively efficient to compute as it involves integration over a finite interval only, and because Gamma function can be readily evaluated using most numerical tools.
We also observe that the coverage is mainly influenced by the LOS BSs. For example, a key term in (\ref{eq:prop1}) is $\lambda\pi R_B^2$ which represents the average number of LOS BSs seen by the user. We now provide analytical expressions for the average harvested power at a connected user.
\begin{proposition}\normalfont
The average harvested power for the connected case ${\bar{P}}_{\textrm{con}}\left(\lambda,\psi\right)$ for an energy outage threshold $\psi\in\left[\psi_\textrm{min},\infty\right)$ is given by
$
{\bar{P}}_{\textrm{con}}\left(\lambda,\psi\right)=
\int_{\psi}^{\infty}P_\textrm{con}\left(\lambda,x\right)\text{d}x+\psi P_\textrm{con}\left(\lambda,\psi\right).
$
\end{proposition}
\begin{proof}
The proof follows by noting that $\gamma$ has nonnegative support, and by treating  ${{P}}_{\textrm{con}}\left(\cdot,\cdot\right)$ as the CCDF of $\gamma$.
\end{proof}
Here, ${{P}}_{\textrm{con}}\left(\cdot,\cdot\right)$ follows from Theorem 1 or Proposition 1.
${\bar{P}}_{\textrm{con}}\left(\lambda,\psi\right)$ can be interpreted as the \emph{useful} average harvested power. This is because only those incident signals that meet the activation threshold can be harvested. 
To get further insights, we now analyze the limiting case $\psi\rightarrow 0$ of Proposition 2. This provides an upper bound on the average harvested power. 
\begin{corollary}\normalfont
The average harvested power for the limiting case $\lim\limits_{\psi\rightarrow 0}{\bar{P}}_{\textrm{con}}\left(\lambda,\psi\right)=
{\bar{P}}_{\textrm{con}}\left(\lambda,0\right)
=\xi\left(\varrho_\los\bar{P}_\los+\varrho_\nlos\bar{P}_\nlos\right)$
, where
\begin{small}
\begin{align}\label{alos}
\bar{P}_\los&=\int\limits_{r_g}^{\infty}\left(P_\textrm{t} M_{\mathrm{t}} M_{\mathrm{r}} C_\los r^{-\alpha_\los}+{\Psi}_\los\left(r\right)+{\Psi}_\nlos\left(\rho_\los(r)\right)\right) 
{\tilde{\tau}}_\los\left(r\right)\text{d}r,
\end{align}
\end{small}
\begin{small}
\begin{align}\label{anlos}
\bar{P}_\nlos&=\int\limits_{r_g}^{\infty}\left(P_\textrm{t} M_{\mathrm{t}} M_{\mathrm{r}} C_\nlos r^{-\alpha_\nlos}+{\Psi}_\los\left(\rho_\nlos(r)\right)+{\Psi}_\nlos\left(r\right)\right)
{\tilde{\tau}}_\nlos\left(r\right)\text{d}r,
\end{align}
\end{small}
\begin{small}
\begin{align}\label{fac3}
{\Psi}_\los\left(x\right)=\kappa C_\los\sum_{i=1}^{4}D_i p_i \int_{x}^{\infty}t^{-(\alpha_\los-1)}p(t)\text{d}t,
\end{align}
\end{small}
\begin{small}
\begin{align}\label{fac4}
{\Psi}_\nlos\left(x\right)=\kappa C_\nlos\sum_{i=1}^{4}D_i p_i \left(\frac{x^{-(\alpha_\nlos-2)}}{\alpha_\nlos-2}-\int_{x}^{\infty}t^{-(\alpha_\nlos-1)}p(t)\text{d}t\right),
\end{align}
\end{small}
and $\kappa= 2\pi\lambda P_\textrm{t}$.
\end{corollary}
\begin{proof}
See Appendix C.
\end{proof}
$\bar{P}_\los$ and $\bar{P}_\nlos$ denote the average harvested power given the user is tagged to an LOS or an NLOS BS. Note that the average harvested power is independent of the small-scale fading parameters. To reveal further insights, we provide the following approximation for the average harvested power (which is validated in Section \ref{secResults1}).
\begin{corollary} The average harvested power for the limiting case, ${\bar{P}}_{\textrm{con}}\left(\lambda,0\right)$, can be further approximated as 
\begin{small}
\begin{align}\label{approx2}
&{\bar{P}}_{\textrm{con}}\left(\lambda,0\right)\overset{(a)}{\approx}
\kappa M_\textrm{t} M_\textrm{r} C_L \int\limits_{r_g}^{R_B}\frac{e^{-\lambda\pi t{^2}}} {t^{\alpha_\mathrm{L}-1}}\text{d}t\nonumber\\
&=\frac{\Gamma\left(1-0.5\alpha_\los;\lambda\pi r_g^2,\infty\right)-\Gamma\left(1-0.5\alpha_\los;\lambda\pi R_B^2,\infty\right)}{{2\left(\kappa M_\textrm{t} M_\textrm{r} C_L\right)^{-1}(\lambda\pi)^{1-0.5 \alpha_\los}}}.
\end{align}
\end{small}
\end{corollary}
\begin{proof}
The proof follows by using the simplifying assumptions of Appendix B, and by further ignoring the contributions from all but the serving LOS BS. 
\end{proof}
This approximation suggests that the average harvested power is mainly determined by the LOS serving link. Note that ${\bar{P}}_{\textrm{con}}(\cdot,\cdot)$ grows linearly with the transmit power $P_\textrm{t}$ since $\kappa=2\pi\lambda P_{\rm{t}}$. Depending on the path loss exponent $\alpha_\mathrm{L}$, it may exhibit a sublinear to approximately-linear scaling with the BS density $\lambda$. When $\alpha_\mathrm{L}$ is large, the denominator $t^{\alpha_\mathrm{L}-1}$ in (a) overshadows the impact of $\lambda$ on the numerator $e^{-\pi\lambda t^2}$. Therefore, the scaling behavior is essentially determined by $\kappa=2\pi\lambda P_\textrm{t}$, which is linear in $\lambda$. This suggests that increasing the transmit power or BS density has almost the same effect on the average harvested power when $\alpha_\mathrm{L}$ is large (e.g., when $\alpha_\mathrm{L}=3$). 
Also note that (\ref{approx2}) is relatively simple as it is expressed in terms of the incomplete Gamma function only. 
\subsubsection*{\textbf{Nonconnected case}}
Having discussed the connected case, we now consider the case where a user operates in the nonconnected mode. The following theorem characterizes the energy coverage probability at a typical user for the nonconnected case.
\begin{theorem}\normalfont
In a mmWave network of density $\lambda$, the energy coverage probability for the nonconnected case $\textrm{P}_\textrm{ncon}\left(\lambda,\psi\right)$ given an outage threshold $\psi$ can be evaluated using
\begin{align}
\textrm{P}_\textrm{ncon}\left(\lambda,\psi\right)\approx
\sum\limits_{k=0}^{N}\left(-1\right)^{k}\binom{N}{k}
e^{-{\overset{}{\Upsilon}}_{k,1}\left(\lambda,\hat{\psi},r_g\right)-{\overset{}{\Upsilon}}_{k,2}\left(\lambda,\hat{\psi},r_g\right)},
\end{align}
where ${\overset{}{\Upsilon}}_{k,1}\left(\cdot\right)$ and ${\overset{}{\Upsilon}}_{k,2}\left(\cdot\right)$ are given by (\ref{fac1}) and (\ref{fac2}) respectively, $\hat{\psi}=\max\left(\frac{\psi}{\xi},\psi_{\min}\right)$, and $r_g$ is the minimum link distance.
\end{theorem}
\begin{proof}
The proof follows from that of Theorem 1 and is therefore omitted.
\end{proof}
Similar to the connected case, the energy coverage probability for this case is also a function of the propagation conditions, the network density and the antenna geometry parameters. We note that the expressions in Theorem 2 are efficient to compute, obviating the need for further simplification. We now consider the average harvested power for the nonconnected case.
\begin{proposition}\normalfont
The average harvested power for the nonconnected case ${\bar{P}}_{\textrm{ncon}}\left(\lambda,\psi\right)$ for an energy outage threshold $\psi\in\left[\psi_\textrm{min},\infty\right)$ is given by
${\bar{P}}_{\textrm{ncon}}\left(\lambda,\psi\right)=
\int_{\psi}^{\infty}P_\textrm{ncon}\left(\lambda,x\right)\text{d}x+\psi P_\textrm{ncon}\left(\lambda,\psi\right).$
\end{proposition}
\begin{proof}
The proof follows from that of Proposition 2 and is therefore omitted.
\end{proof}
\begin{corollary}\normalfont
The average harvested power for the limiting case $\lim\limits_{\psi\rightarrow 0}{\bar{P}}_{\textrm{ncon}}\left(\lambda,\psi\right)={\bar{P}}_{\textrm{ncon}}\left(\lambda,0\right)$
is given by
\begin{align}\label{equation19}
{\bar{P}}_{\textrm{ncon}}\left(\lambda,0\right)=
\xi\left({\Psi}_\los\left(r_g\right)+{\Psi}_\nlos\left(r_g\right)\right),
\end{align}
where ${\Psi}_\los\left(\cdot\right)$ and ${\Psi}_\nlos\left(\cdot\right)$ are given in (\ref{fac3}) and (\ref{fac4}) respectively.
\end{corollary}
\begin{proof}
The proof follows from that of Corollary 1 and is therefore omitted.
\end{proof}
The average harvested power for the nonconnected case scales \textit{linearly} with the transmit power and the BS density.  This follows from (\ref{equation19}) as both ${\Psi}_\los\left(\cdot\right)$ and ${\Psi}_\nlos\left(\cdot\right)$ relate linearly with the transmit power and density via the term $\kappa=2\pi\lambda P_\textrm{t}$. This also suggests that increasing the transmit power or density has the same effect on the average harvested power. Note that this is different from the connected case where the path loss exponent affects how average harvested power scales with the BS density.
\subsection{Results and Design Insights}\label{secResults1}
We first verify the accuracy of the analytical expressions presented in Section \ref{secStoch} using simulations. 
We then study how key design parameters such as the antenna beam pattern affects the energy coverage probability in purely connected ($\epsilon\rightarrow 1$) and nonconnected ($\epsilon\rightarrow 0$) networks.
We also compare the performance of mmWave energy harvesting with lower frequency solutions.
After developing key insights for purely connected/nonconnected scenarios, we provide energy coverage results for the general case ($0<\epsilon<1$), where the network serves both types of users.
\subsubsection*{Validation}
In the following plots, the users are assumed to be equipped with a single omnidirectional receive antenna, the mmWave carrier frequency is set to 28 GHz, the blockage constant $\beta=0.0071$~\cite{bai2015}, and $\psi>\psi_\textrm{min}$. In other words, for a given $\psi$, the plots are valid for any $\psi_\textrm{min}<\psi$. Note that for the less relevant case when $\psi<\psi_\textrm{min}$, the energy coverage probability flattens out, and is specified by $\textrm{P}_\textrm{con}\left(\lambda,\psi_\textrm{min}\right)$ or $\textrm{P}_\textrm{ncon}\left(\lambda,\psi_\textrm{min}\right)$. 
 Without loss of generality, we set the rectifier efficiency $\xi=1$ since this parameter does not impact the shape of the results, i.e., setting $\xi<1$ results in shifting all the curves to the left by the same amount. {We assume the rectifier efficiency to be the same when comparing mmWave and UHF. Note that there are no standard values for $\xi$ since prior work has reported widely varying values\cite{valenta2014,tabesh2015power} depending on the device technology, operating frequency, etc. For example, \cite{tabesh2015power,Charthad2016mmWaveWPT} suggest that a mmWave energy harvesting circuit may have better overall performance than its lower frequency counterparts.}  
Fig. \ref{fig:Con_Ant} and \ref{fig:Con_Den} plot the energy coverage probability for the connected case using different model parameters. The analytical results based on Theorem 1 are obtained using $N=5$ terms in the approximation. The simulation results are generated using Monte Carlo simulations with 10,000 runs. Similarly, using Theorem 2, Fig. \ref{fig:Uncon_Ant} and \ref{fig:Uncon_Den} plot the energy coverage probability for the nonconnected case. There is a nice agreement between analytical and simulation results.
\subsubsection*{Connected case ($\epsilon\rightarrow 1$)}
In Fig. \ref{fig:Con_Ant}, we plot the energy coverage probability with three distinct transmit beam patterns for a given network density. 
We observe that the energy harvesting performance improves with narrower beams, i.e., smaller beamwidths and larger directivity gains.
As the beamwidth decreases, relatively fewer beams from the neighboring BSs would be incident on a typical user. But the beams that do reach, will have larger directivity gains, resulting in an overall performance improvement.
This is possible due to the use of potentially large antenna arrays at the mmWave BSs. Note that this performance boost will possibly be limited due to the ensuing EIRP (equivalent isotropically radiated power) or other safety regulations on future mmWave systems \cite{wu2015human}.

For the purpose of comparison, we also plot the energy coverage probability for UHF energy harvesting under realistic assumptions. Given the current state-of-the-art \cite{ghosh2010fundamentals,rappaport2014millimeter}, the UHF BSs are assumed to have 8 transmit antennas each. Further, they are assumed to employ maximal ratio transmit beamforming to serve a connected user. For the channel model, we assume an IID Rayleigh fading environment and a path loss exponent of $3.6$ (no blockage is considered). The network density is set to $25$ nodes/$\textrm{km}^2$, which corresponds to an \emph{average} distance of about $113$ m to the closest UHF BS. The carrier frequency is set to $2.1$ GHz and the transmission bandwidth is 100 MHz.
As can be seen from Fig. \ref{fig:Con_Ant}, mmWave energy harvesting could provide considerable performance gain over its lower frequency counterpart. Moreover, the anticipated dense deployments of mmWave networks would further widen this gap. This effect is illustrated in Fig. \ref{fig:Con_Den}, where we plot the energy coverage probability for different mmWave network densities for a given transmit antenna beam pattern.
In Fig. \ref{fig:Con_Rate}, we use Proposition 2 to plot the average harvested power at a typical mmWave user against the transmit array size. The plots based on Corollary 1 are also included. 
We note that the limiting case $\psi\rightarrow 0$ treated in Corollary 1 closely approximates the average harvested power obtained using Proposition 2.
This figure also confirms our earlier intuition that mmWave energy harvesting can benefit from (i) potentially large antenna arrays at the BSs, and (ii) high BS density, which would be the key ingredients of future mmWave cellular systems. 
Fig. \ref{fig:Con_Rate2} shows how the path loss exponent impacts the scaling behavior of the average harvested power with BS density, corroborating the discussion following Corollary 2.

\begin{figure*}[t]
	\centering
		\subfloat[][]{%
			\centering
			\includegraphics[width=3.3in]{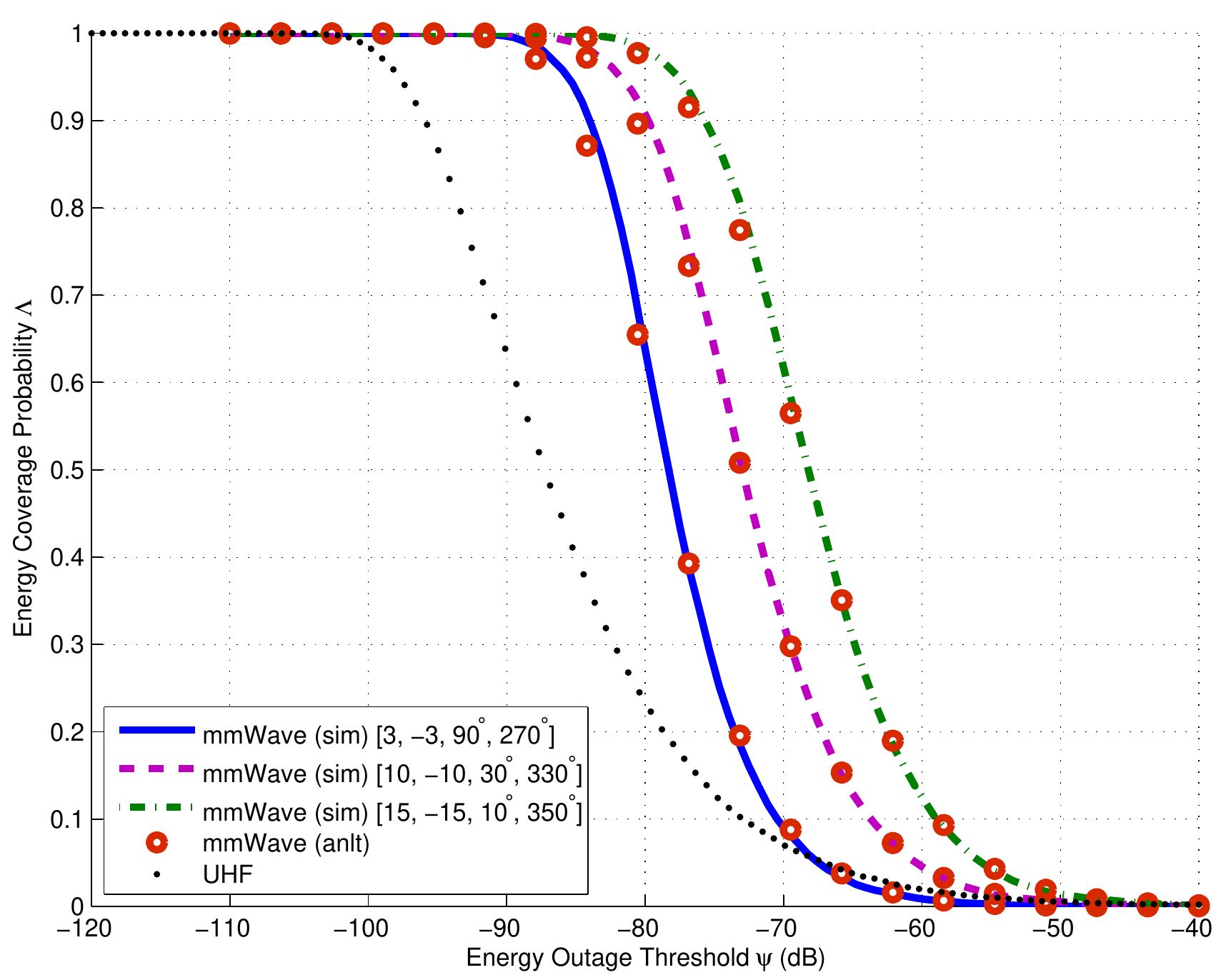}
			\label{fig:Con_Ant}} \hfill
		\subfloat[][]{
			\centering
			\includegraphics[width=3.3in]{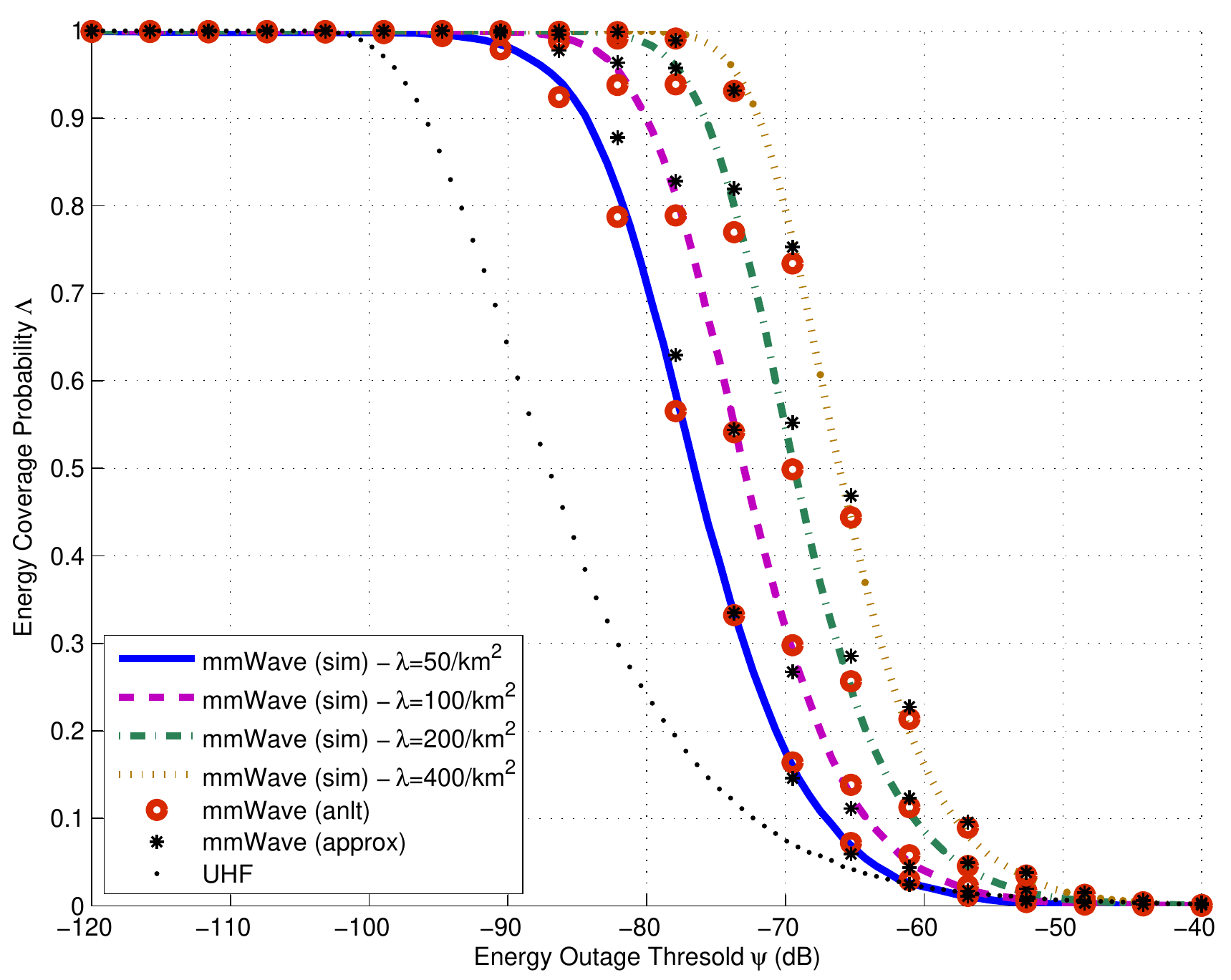}
	    	\label{fig:Con_Den}}	
\caption{(a) Energy coverage probability $\Lambda(\epsilon,\psi,\lambda)$ for different transmit antenna beam patterns parameterized by $[M_\textrm{t}, m_\textrm{t}, \theta_\textrm{t}, \bar{\theta}_\textrm{t}]$ in a purely connected network ($\epsilon=1$, $\lambda=100/{\text{k}{\text{m}}^2}$). The performance improves with narrower beams for this case. $P_\textrm{t}=13$ dB, $W= 100$ MHz, $\alpha_\los=2$, $\alpha_\nlos=4$, $N_\los=2$, $N_\nlos=3$, and $r_g=1$ m. There is a nice agreement between Monte Carlo simulation (sim) results and the analytical (anlt) results obtained using Theorem 1 with $N=5$ terms. (b) Energy coverage probability $\Lambda(1,\psi,\lambda)$ for different network densities for connected users. Transmit beam pattern is fixed to $[10, -10, 30^\circ, 330^\circ]$. Other parameters are the same as given in Fig. \ref{fig:Con_Ant}. 
Also included are the results based on the analytical approximation (approx) in Proposition 1. The approximation becomes tighter as the density increases.}	
\end{figure*}

\begin{figure*}[t]
	\centering
		\subfloat[][]{%
			\centering
			\includegraphics[width=3.3in]{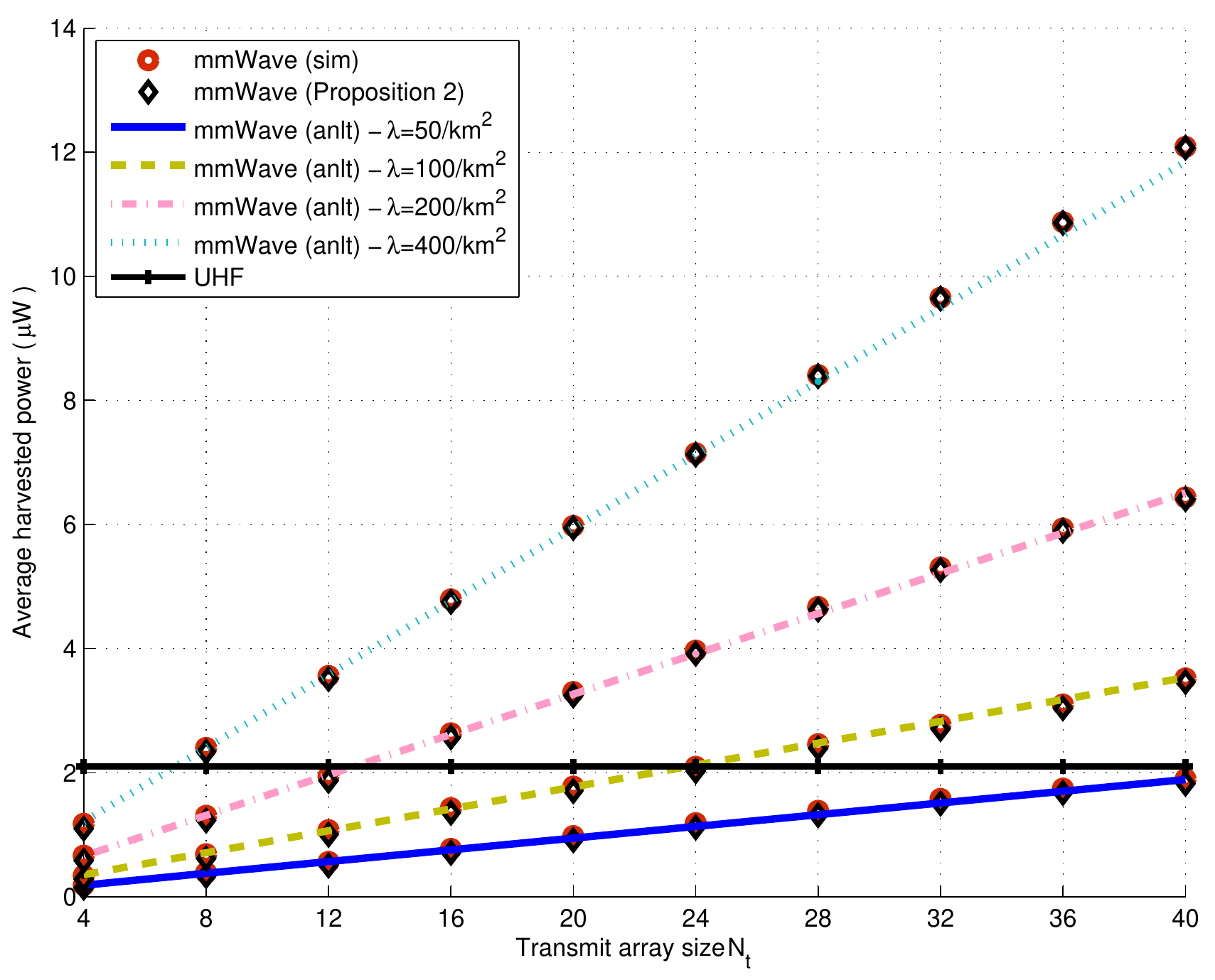}
			\label{fig:Con_Rate}} \hfill
		\subfloat[][]{
			\centering
			\includegraphics[width=3.3in]{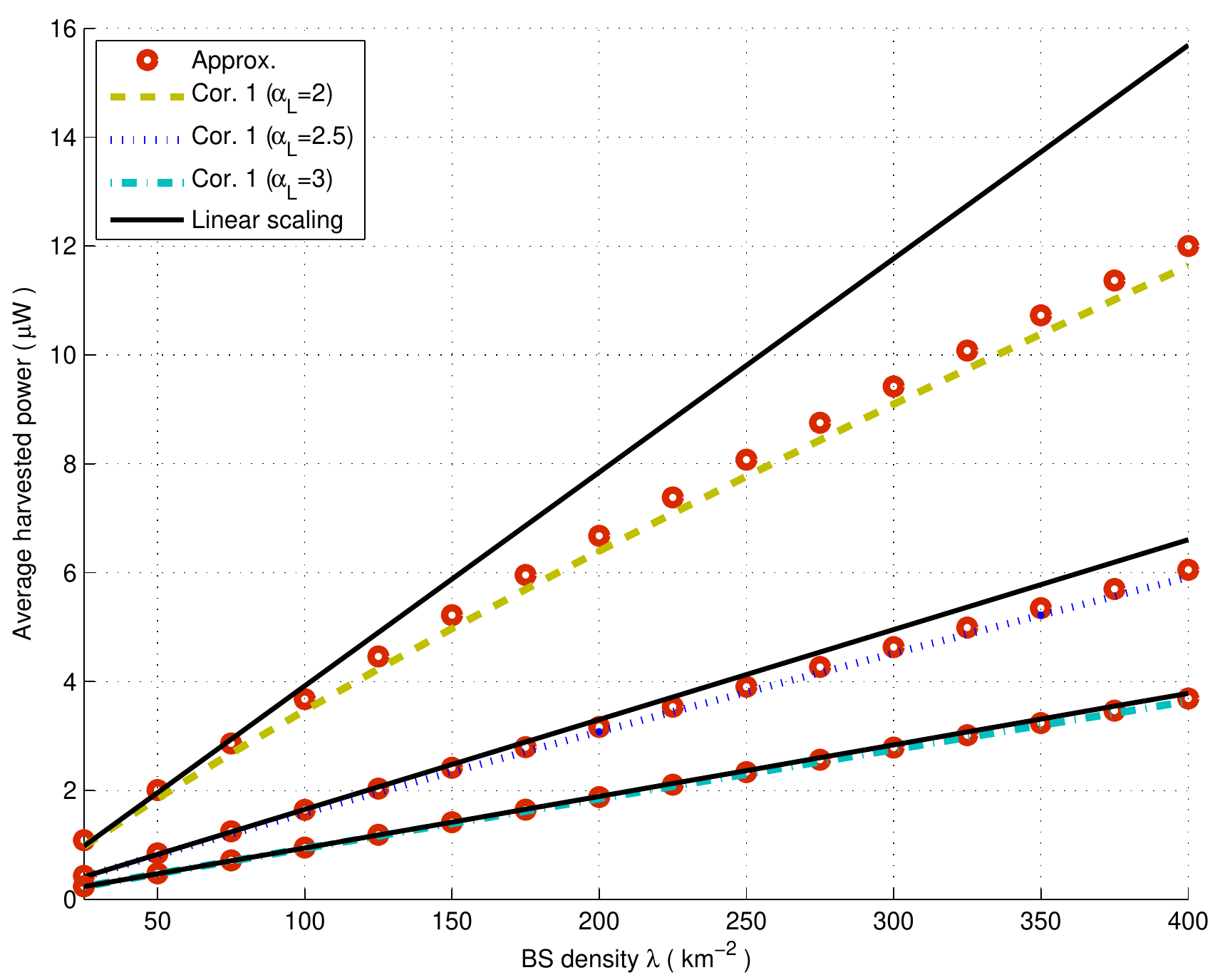}
	    	\label{fig:Con_Rate2}}	
\caption{(a) The average harvested power in a connected mmWave network for different number of BS antennas $N_\textrm{t}$ and deployment densities $\lambda$. 
Results based on Proposition 2 are obtained for $\psi=-35$ dBm. The analytical (anlt) results based on Corollary 1 ($\psi\rightarrow 0$) are validated using Monte Carlo simulations (sim); and closely approximate the average harvested energy obtained using Proposition 2. The transmit antenna beam patterns are calculated using the approximations used for obtaining Fig. \ref{fig:tradeoff}. Other simulation parameters are the same as used in Fig. \ref{fig:Con_Ant}. For comparison, a plot for a UHF system is also included. (b) Plots the average harvested power (Corollary 1) vs. BS density for $N_\textrm{t}=32$. The plot validates the approximation in (\ref{approx2}). Moreover, it shows how the average power scales with the BS density for different path loss exponents $\alpha_\mathrm{L}$. For illustration, also included are the solid lines corresponding to the (hypothetical) case when average power scales linearly with density. The scaling tends to become linear as $\alpha_\mathrm{L}$ is increased.}	
\end{figure*}
\subsubsection*{Nonconnected case ($\epsilon\rightarrow 0$)}
We now analyze the energy harvesting performance when the harvesting devices operate in the nonconnected mode.
In a stark contrast to the connected case, Fig. \ref{fig:Uncon_Ant} shows that for the nonconnected case, mmWave energy harvesting could benefit from using wider beams.
This is because BS connectivity (alignment) is critical for the nonconnected case. With wider beams, it is more likely that a mmWave BS gets aligned with a receiver, albeit at the expense of the beamforming gain.
Furthermore, a comparison with UHF energy harvesting shows that mmWave energy harvesting gives a comparable performance to its UHF counterpart.
Similarly, Fig. \ref{fig:Uncon_Den} plots the energy coverage probability for different deployment densities. We note that performance can be substantially improved with denser deployments, which would be a key feature of future mmWave cellular systems.

\begin{figure*}[t]
	\centering
		\subfloat[][]{%
			\centering
			\includegraphics[width=3.3in]{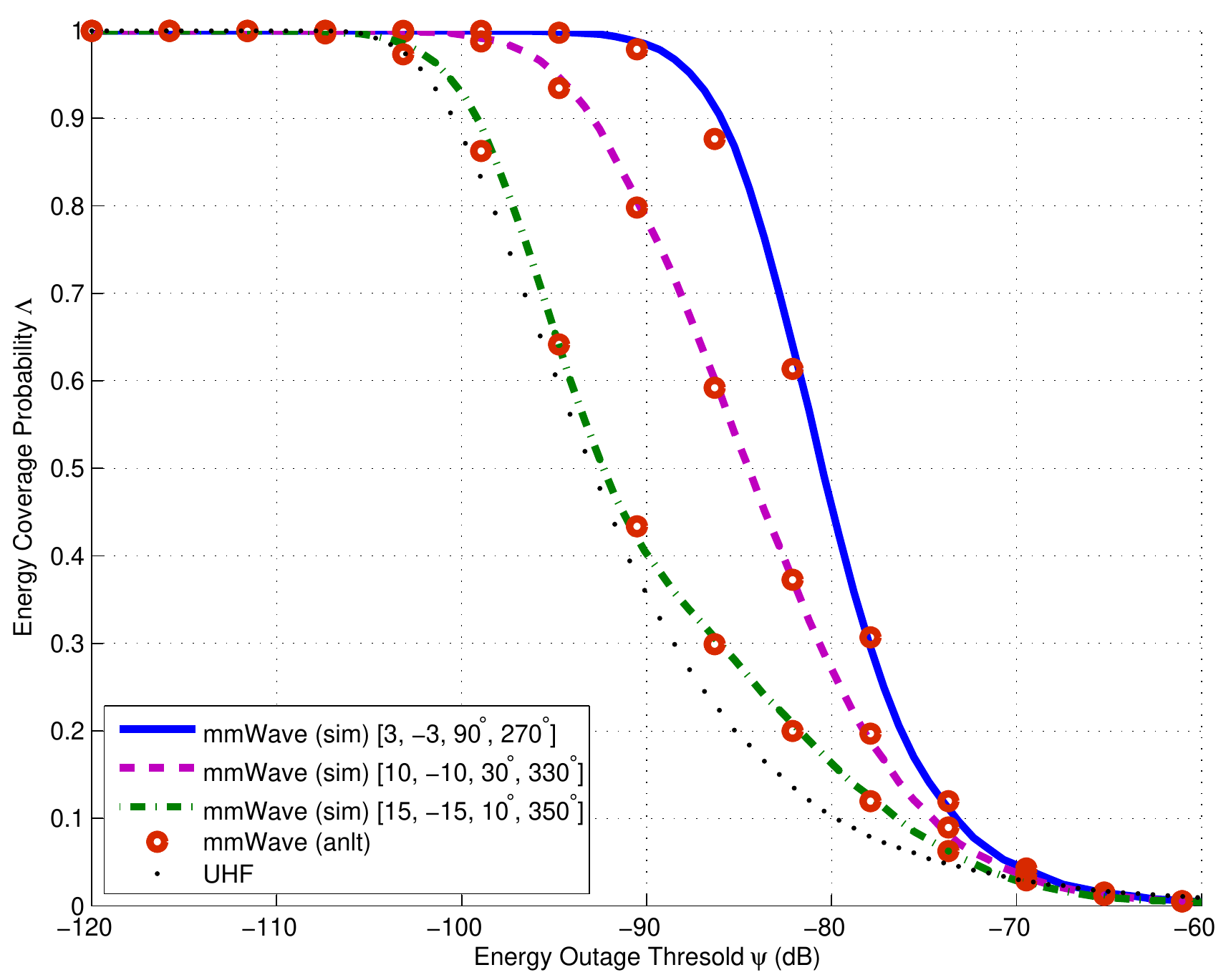}
			\label{fig:Uncon_Ant}} \hfill
		\subfloat[][]{
			\centering
			\includegraphics[width=3.3in]{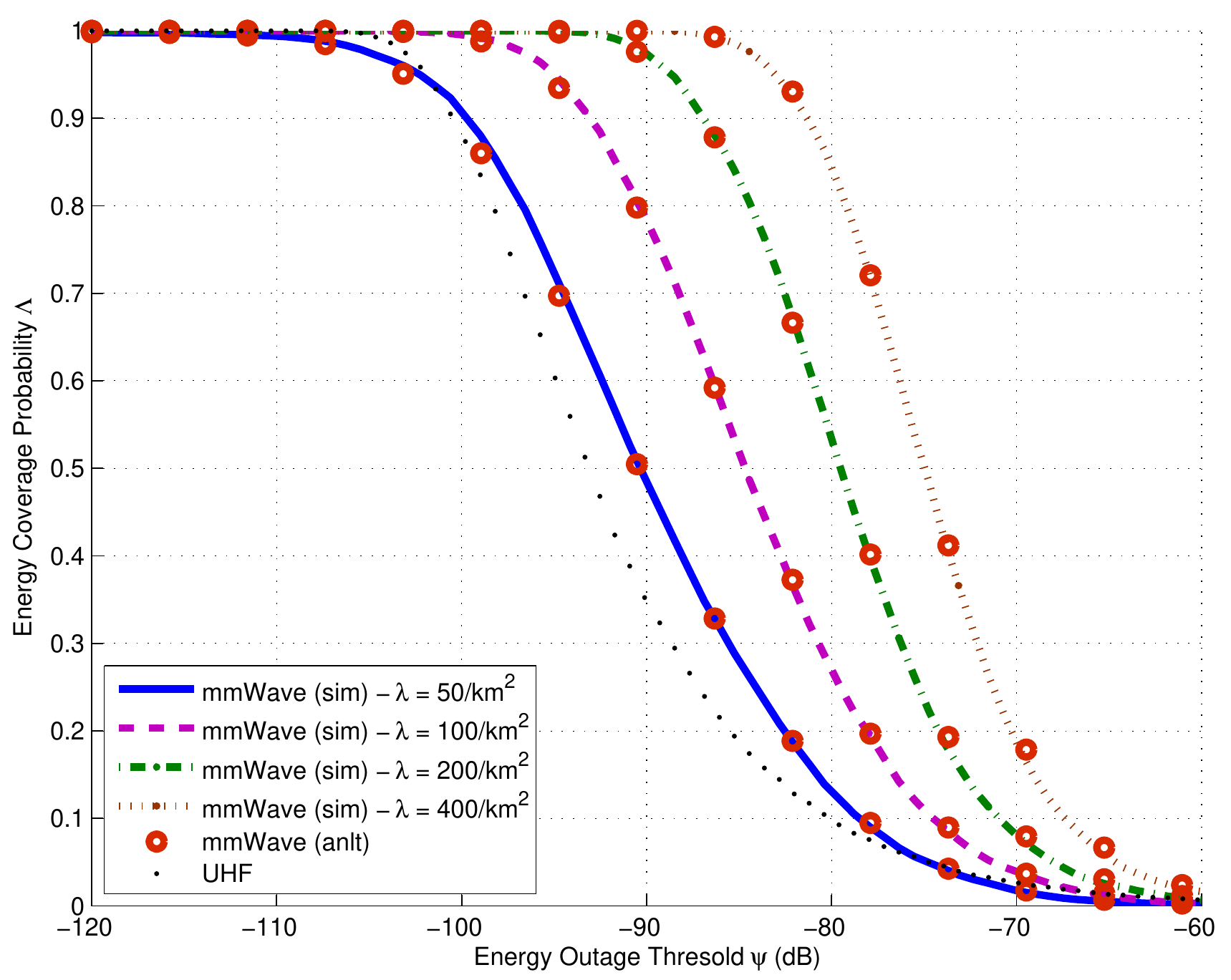}
	    	\label{fig:Uncon_Den}}	
\caption{(a) Energy coverage probability $\Lambda(\epsilon,\psi,\lambda)$ for different transmit antenna beam patterns in a nonconnected network ($\epsilon=0$, $\lambda=100/{\text{k}{\text{m}}^2}$). The performance improves with wider beams for this case. Other simulation parameters are same as given in Fig. \ref{fig:Con_Ant}. Monte Carlo simulation (sim) results validate the analytical (anlt) results obtained using Theorem 2 with $N=5$ terms. (b) Energy coverage probability $\Lambda(0,\psi,\lambda)$ for different network densities for nonconnected users. Transmit beam pattern is fixed to $[10, -10, 30^\circ, 330^\circ]$. Other parameters are same as given in Fig. \ref{fig:Con_Ant}. Monte Carlo simulation (sim) results validate the analytical (anlt) results obtained using Theorem 2 with $N=5$ terms.}	
\end{figure*}
\subsubsection*{General case ($0<\epsilon<1$)}
Having presented the energy coverage trends for the two extreme network scenarios, we now consider the general case where the user population consists of both connected and nonconnected users.
We expect this to be the likely scenario for reasons explained in the network model (Section \ref{secNet}).
As described in Section \ref{secAnt}, an antenna beam pattern can be characterized by the half power beamwidth and directivity gain for both the main and side lobes. By tuning these parameters, the beam pattern can be particularized to a given antenna array.
As an example, we assume that uniform linear arrays (ULA) are deployed at the mmWave BSs. We use the following relations to approximate the main and side lobe beamwidths as a function of the transmit array size:
$
\theta_\textrm{t}\approx\frac{360}{\pi}\arcsin\left(\frac{0.892}{N_\textrm{t}}\right)
$ and
$\bar{\theta}_\textrm{t}\approx
\frac{720}{\pi}\left|\arcsin\left(\frac{2}{N_\textrm{t}}\right)\right|
$ \cite{van2004detection}. We use $M_\textrm{t}=10V\log\left(N_\textrm{t}\right)$ and $m_{\textrm{t}}=V\left(M_\textrm{t}-12\right)$ for the directivity gains of the main and side lobes\cite{van2004detection}.
To ensure the power normalization, the constant $V$ is chosen to satisfy $\frac{\theta_\textrm{t}}{2\pi}M_{\textrm{t}}+\frac{\bar{\theta}_\textrm{t}}{2\pi}m_{\textrm{t}}=1$.

In Fig. \ref{fig:tradeoff}, we plot the overall energy coverage probability $\Lambda(\epsilon,\psi,\lambda)$ against transmit array size $N_\textrm{t}$ for different values of parameter $\epsilon$. We find that the optimal transmit array size depends on the statistics of the user population. For example, when $\epsilon$ is large, it is desirable to use large antenna arrays at the BSs. When $\epsilon$ is small, it is favorable to use small antenna arrays to improve the overall energy coverage probability.
Depending on the network load (or the user population \emph{mix}) captured via $\epsilon$, the energy coverage probability can be substantially improved by intelligent antenna switching schemes.
Since the parameter $\epsilon$ would typically vary over large time-scales, such schemes would be practically feasible.
\begin{figure} [h]
\centerline{
\includegraphics[width=0.9\columnwidth]{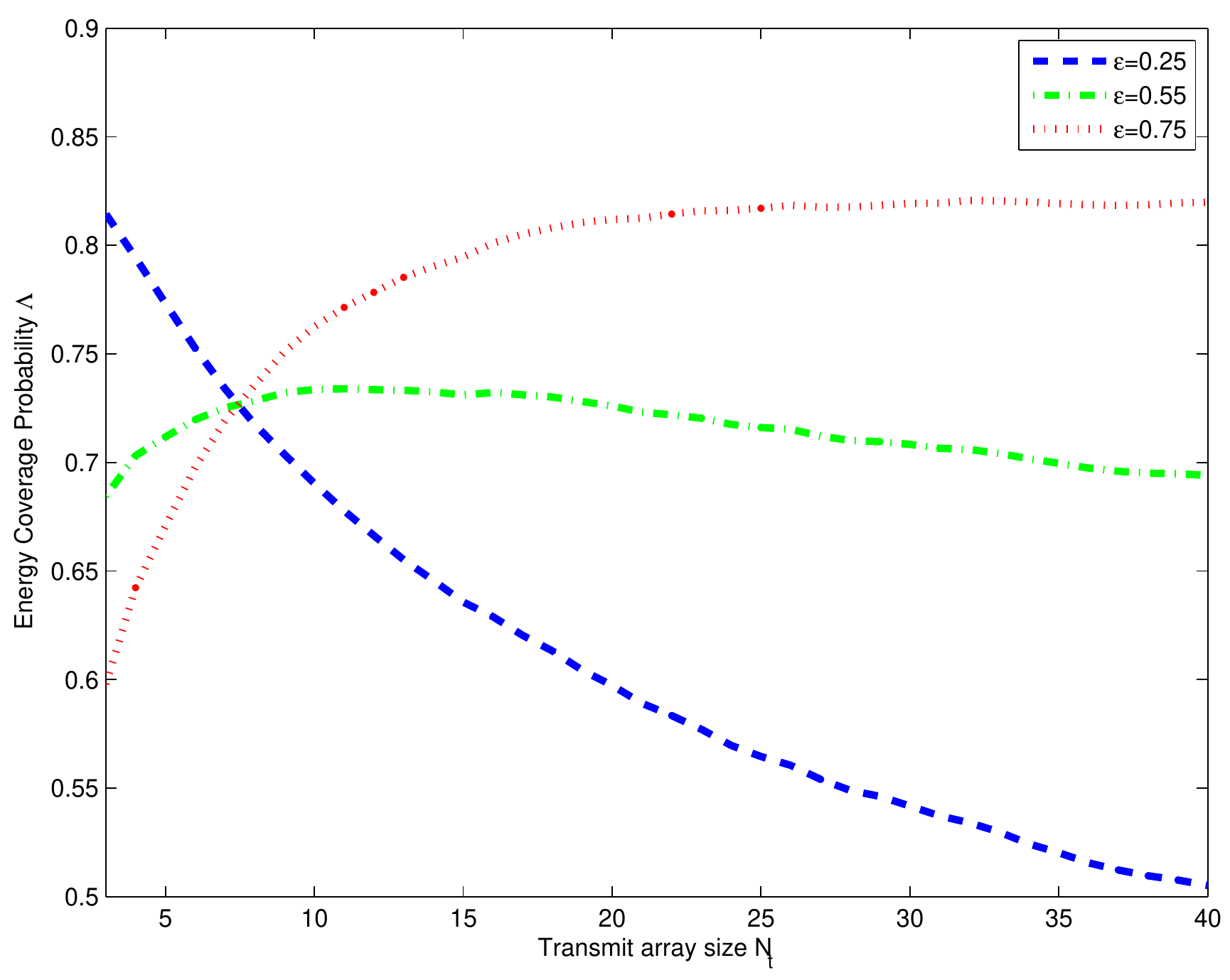}
}
\caption{The overall energy coverage probability $\Lambda(\epsilon,\psi,\lambda)$ for different values of $\epsilon$. 
Depending on the fraction of users operating in connected/nonconnnected modes, the transmit array size (which controls the beamforming beamwidth in this example) can be optimized to maximize the network-wide energy coverage. This could translate into massive gains given that the number of served devices would be potentially large.
The users are assumed to be equipped with a single omnidirectional receive antenna. The energy outage threshold $\psi$ is $-70$ dB for $\Phi_{u,\textrm{con}}$ and $-85$ dB for $\Phi_{u,\textrm{ncon}}$. $P_\textrm{t}=13$ dB, $\lambda=200/\textrm{km}^2$. Channel parameters are the same as used in Fig. \ref{fig:Con_Ant}.}
\label{fig:tradeoff}
\end{figure}

Having presented the energy coverage trends for mmWave energy harvesting, we now consider the scenario where the user attempts to extract both energy and information from the incident mmWave signals.
\section{MmWave Simultaneous Information and Power Transfer}\label{secSWIPT}
In this section, we consider the case where the energy harvesting device also attempts to decode information from the received signals, in what is known as simultaneous wireless information and power transfer (SWIPT)\cite{EnergyHarvestWirelessCommSurvey2015,WPCSurvey2015}. We now assume that the energy harvesting receiver is also equipped with an information decoding circuit.
We focus on the case where a given user is already aligned with its serving BS, i.e., $\epsilon=1$ for this section.
Further, we consider a power splitting receiver architecture\cite{WPCSurvey2015} where the received signal is split using factors $\sqrt{\nu}$ and $\sqrt{1-\nu}$, $\nu\in[0,1]$. A fraction $\sqrt{1-\nu}$ of received signal is available for energy harvesting, while the remaining signal is used for information decoding.
With this notation, the signal-to-interference-plus-noise ratio (SINR) at a typical receiver can be expressed as $\textrm{SINR}=\frac{\nu S}{\nu(I+\sigma^2)+\sigma_c^2}$, where $S=
P_\textrm{t}M_\textrm{t}M_\textrm{r} H_0 g_0(r_0)$ denotes the useful signal power and $I=\sum_{\ell>0,\ell\in\Phi(\lambda)\setminus\mathbb{B}(r_g)}^{}
P_\textrm{t}\delta_\ell H_\ell g_\ell(r_\ell)$ gives the aggregate interference power from the neighboring BSs. $\sigma^2$ is the thermal noise power before splitting, while $\sigma_c^2$ captures possible signal degradation after power splitting. Similarly, $\gamma=\left(1-\nu\right)\xi\left(S+I+\sigma^2\right)\mathbbm{1}_{\{{S+I+\sigma^2>\psi_{\rm{min}}}\}}$ denotes the
received signal power fed to the energy harvester.
Note that a user will be in outage if the harvested energy and/or the SINR fall below their respective thresholds.
We now define $P_\text{suc}(\lambda,T,\psi,\nu)=\Pr\left[\text{SINR}>T,\gamma>\psi\right]$ to be the probability of successful reception given the SINR outage threshold $T$, the energy outage threshold $\psi$, and the power splitting ratio $\nu$.
Extending the results from the previous sections, we now provide an analytical expression to characterize the system performance with SWIPT.

\subsection{Stochastic Geometry Analysis} \label{subsec:Analysis_SWIPT}
Before stating the main result of this section, we first provide a lemma for the SINR coverage probability at a mmWave receiver \cite{bai2015}.
\begin{lemma}[{Modified from \cite[Theorem 1]{bai2015}}]\normalfont
In a mmWave network of density $\lambda$, the SINR coverage probability $\textrm{P}_\textrm{cov}\left(\lambda,T,\nu\right)$ at a SWIPT device, given an SINR outage threshold $T$ and a power splitting ratio $\nu$, is given by
\begin{align}
\textrm{P}_\textrm{cov}\left(\lambda,T,\nu\right)&=
\textrm{P}_\textrm{cov,L}\left(\lambda,T,\nu\right)\varrho_\textrm{L}+\textrm{P}_\textrm{cov,N}\left(\lambda,T,\nu\right)\varrho_\textrm{N},
\end{align}
where $\varrho_\textrm{L}=1-\varrho_\textrm{N}$ is defined in Lemma 2, and $\textrm{P}_\textrm{cov,L}(\cdot)$ gives the conditional SINR coverage probability given the device is served by a LOS BS, and can be approximated as
\begin{small}
\begin{align}\label{los2}
\textrm{P}_\textrm{cov,L}\left(\lambda,T,\nu\right)&\approx 
\sum\limits_{k=1}^{N_\los}\left(-1\right)^{k+1}\binom{N_\los}{k}
\int\limits_{r_g}^{\infty}
e^{-\frac{kc_\los {r}^{\alpha_\los}T \left(\sigma^2+{\nu^{-1}\sigma^2_c}\right)}{PC_\los M_\textrm{t}M_\textrm{r}}}\nonumber\\
\hspace{0.5 in}\times
& \quad e^{-{\overset{}{\Delta}}_{k,1}\left(T,r\right) -{\overset{}{\Delta}_{k,2}\left(T,r\right)}}\tilde{\tau}_\los\left(r\right)\text{d}r.
\end{align}
\end{small}Similarly, the conditional SINR coverage probability for the NLOS case $\textrm{P}_\textrm{cov,N}(\cdot)$ is given by
\allowdisplaybreaks
\begin{small}
\begin{align}\label{nlos2}
\textrm{P}_\textrm{cov,N}\left(\lambda,T,\nu\right)&\approx 
\sum\limits_{k=1}^{N_\nlos}\left(-1\right)^{k+1}\binom{N_\nlos}{k}
\int\limits_{r_g}^{\infty}
e^{-\frac{kc_\nlos r^{\alpha_\nlos}T \left(\sigma^2+{\nu^{-1}\sigma^2_c}\right)}{PC_\nlos M_\textrm{t}M_\textrm{r}}}\nonumber\\
& \times \hspace{0.1 in}e^{-{\overset{}{\Delta}}_{k,3}\left(T,r\right)-{\overset{}{\Delta}_{k,4}\left(T,r\right)}}\tilde{\tau}_\nlos\left(r\right)\text{d}r,
\end{align}
\end{small}
where
\begin{small}
\begin{align}
{\overset{}{\Delta}_{k,1}}\left(T,x\right)=
2\pi\lambda\sum_{i=1}^{4}p_i \int\limits_{x}^{\infty}&\left(1-{\left[1+\frac{c_\los k\tilde{D}_iTx^{\alpha_\los}} {{N_\los} t^{\alpha_\los}}\right]^{-N_\los}}\right)\nonumber\\
&\hspace{0.5in}\times p(t)t\text{d}t,
\end{align}
\end{small}
\begin{small}
\begin{align}
{\overset{}{\Delta}_{k,2}}\left(T,x\right)=
2\pi\lambda\sum_{i=1}^{4}p_i \int\limits_{\rho_{\los}(x)}^{\infty}&\left(1-{\left[1+\frac{c_\los k\tilde{D}_iC_\nlos T x^{\alpha_\los}} {{N_\nlos}C_\los t^{\alpha_\nlos}}\right]^{-N_\nlos}}\right)\nonumber\\
&\hspace{0.5in}\times \left(1-p(t)\right)t\text{d}t,
\end{align}
\end{small}
\begin{small}
\begin{align}
{\overset{}{\Delta}_{k,3}}\left(T,x\right)=
2\pi\lambda\sum_{i=1}^{4}p_i \int\limits_{\rho_{\nlos}(x)}^{\infty}&\left(1-{\left[1+\frac{c_\nlos k\tilde{D}_iC_\los T x^{\alpha_\nlos}} {{N_\los}C_\nlos t^{\alpha_\los}}\right]^{-N_\los}}\right)
\nonumber\\
&\hspace{0.5in}\times p(t)t\text{d}t,
\end{align}
\end{small}
\begin{small}
\begin{align}
{\overset{}{\Delta}_{k,4}}\left(T,x\right)&=
2\pi\lambda\sum_{i=1}^{4}p_i \int\limits_{x}^{\infty}\left(1-{\left[1+\frac{c_\nlos k\tilde{D}_i T x^{\alpha_\nlos}} {{N_\nlos} t^{\alpha_\nlos}}\right]^{-N_\nlos}}\right)\nonumber\\
&\hspace{1in}\times \left(1-p(t)\right)t\text{d}t,
\end{align}
\end{small}$\tilde{D}_i=\frac{D_i}{M_\textrm{t} M_\textrm{r}}$ for $i\in\{1,2,3,4,5\}$, $c_\los=N_\los\left(N_\los!\right)^{-\frac{1}{N_\los}}$ and $c_\nlos=N_\nlos\left(N_\nlos!\right)^{-\frac{1}{N_\nlos}}$.
\label{lem4}
\end{lemma}
The following theorem provides the main analytical result of this section.
\begin{theorem}\normalfont
In a mmWave network of density $\lambda$, the success probability $P_\text{suc}(\lambda,T,\psi,\nu)$ given the SINR outage threshold $T$, the energy outage threshold $\psi$, and the power splitting ratio $\nu$ is given by
\begin{small}
\begin{align}\label{eq:thm3}
P_{\text{suc}}\left(\lambda,T,\psi,\nu\right)&\approx P_{\text{cov}}\left(\lambda,T,\nu\right) \tilde{P}_{\text{con}}\left(\lambda,\mu\right)\nonumber\\
&\qquad\quad+P_{\text{con}}\left(\lambda,\varphi\right)\left[1-\tilde{P}_{\text{con}}\left(\lambda,\mu\right)\right],
\end{align}
\end{small}where the SINR coverage probability $P_{\textrm{cov}}(\cdot)$ can be evaluated using the expressions in Lemma \ref{lem4}, while the energy coverage probability $P_{\textrm{con}}(\cdot)$ follows from Theorem 1. We further define ${\tilde{P}}_{\textrm{con}}(\lambda,\mu)
=\tilde{P}_{\textrm{con},\los}(\lambda,\mu)\varrho_\los+{\tilde{P}}_{\textrm{con},\nlos}(\lambda,\mu) \varrho_\nlos$, where $\tilde{P}_{\textrm{con},\los}(\cdot)$ and $\tilde{P}_{\textrm{ncon},\nlos}(\cdot)$ are specified by (\ref{los}) and (\ref{nlos}) respectively, by setting $\zeta_k^\los(\cdot)=\zeta_k^\nlos(\cdot)=1$. Moreover, $\mu$ and $\varphi$ depend on several parameters including the power splitting ratio $\nu$, the SINR outage threshold $T$, the energy outage threshold $\psi$, the harvester activation threshold $\psi_\textrm{min}$, the rectifier efficiency $\xi$, and the noise parameters. Further,
$\mu=\frac{\hat{\psi}}{(1-\nu)(1+T)}-\sigma^2-\frac{\sigma_c^2}{\nu(1+\frac{1}{T})},$
$\varphi=\frac{\hat{\psi}}{\left(1-\nu\right)},$ and $\hat{\psi}=\max\left(\frac{\psi}{\xi},\psi_\textrm{min}\right)$.
\end{theorem}
\begin{proof}
See Appendix D.
\end{proof}
Note that $\tilde{P}_{\textrm{con}}(\lambda,\mu)$ in (\ref{eq:thm3}) is the interference CCDF evaluated at parameter $\mu$. It plays a key role in determining the operating mode of the system. Though the interference is harmful for information decoding, it can be beneficial for energy harvesting.
When the interference is high, the SINR coverage probability will typically limit the success probability. In the other extreme, the energy coverage probability will play the limiting role.
Also note that the success probability can be optimized over the design paramter $\nu$, given other parameters. Moreover, we can recover Theorem 1 and Lemma 4 from (\ref{eq:thm3}) by letting $\psi\rightarrow0$ and $T\rightarrow0$, respectively.

Note that, in principle, the success probability at a connected mmWave energy harvesting or SWIPT device can be further improved by leveraging large antenna arrays at the receiver, thanks to smaller wavelengths.
Though our analytical model allows the users to have receive antenna arrays, it implicitly assumes the presence of ideal RF combining circuitry consisting of power-hungry components such as phase shifters, multiple RF chains, etc. 
When large antenna arrays are used at the receiver, the power consumption due to additional antenna circuitry may get prohibitively high, overshadowing the array gains. 
As SWIPT typically targets low-power devices, we present a simple low-power receiver architecture in the next section. 
Note that the analytical results based on Theorem 3 can be interpreted as an upper bound on performance when the receiver consists of suboptimal components (as is the case in the following section).

\subsection{Low-power Receiver Architecture}
We now propose a novel architecture for a mmWave SWIPT receiver with multi-antenna array, as depicted in Fig. \ref{fig:Model_Switches_v1}. In this architecture, we assume per-antenna power splitting with parameter $\nu$ (as defined earlier). After power splitting, the input signal at each antenna passes through a rectifier, followed by a DC combiner that yields the harvested energy.
For the information path, after passing through power splitters, the received signals are first combined in the RF domain using a combining vector $\bw$. The resulting signal is then decoded in the baseband.
Because they require extremely small power, the combining vector is assumed to be implemented using switches \cite{Mendez-Rial2015}, i.e., $\bw={[\textrm{w}_1,\cdots,\textrm{w}_{N_\mathrm{r}}]}^*\in[0,1]^{N_\textrm{r}}$. 
Note that both the signals for information decoding and energy harvesting are in the order of $\mu$W (Fig. \ref{fig:Cov_3D}). It is worth mentioning that recent results have shown that mmWave energy harvesting circuits can run with only a few $\mu$W~\cite{tabesh2015power,Charthad2016mmWaveWPT}.

\begin{figure} [t]
	\centerline{
		\includegraphics[width=0.9\columnwidth]{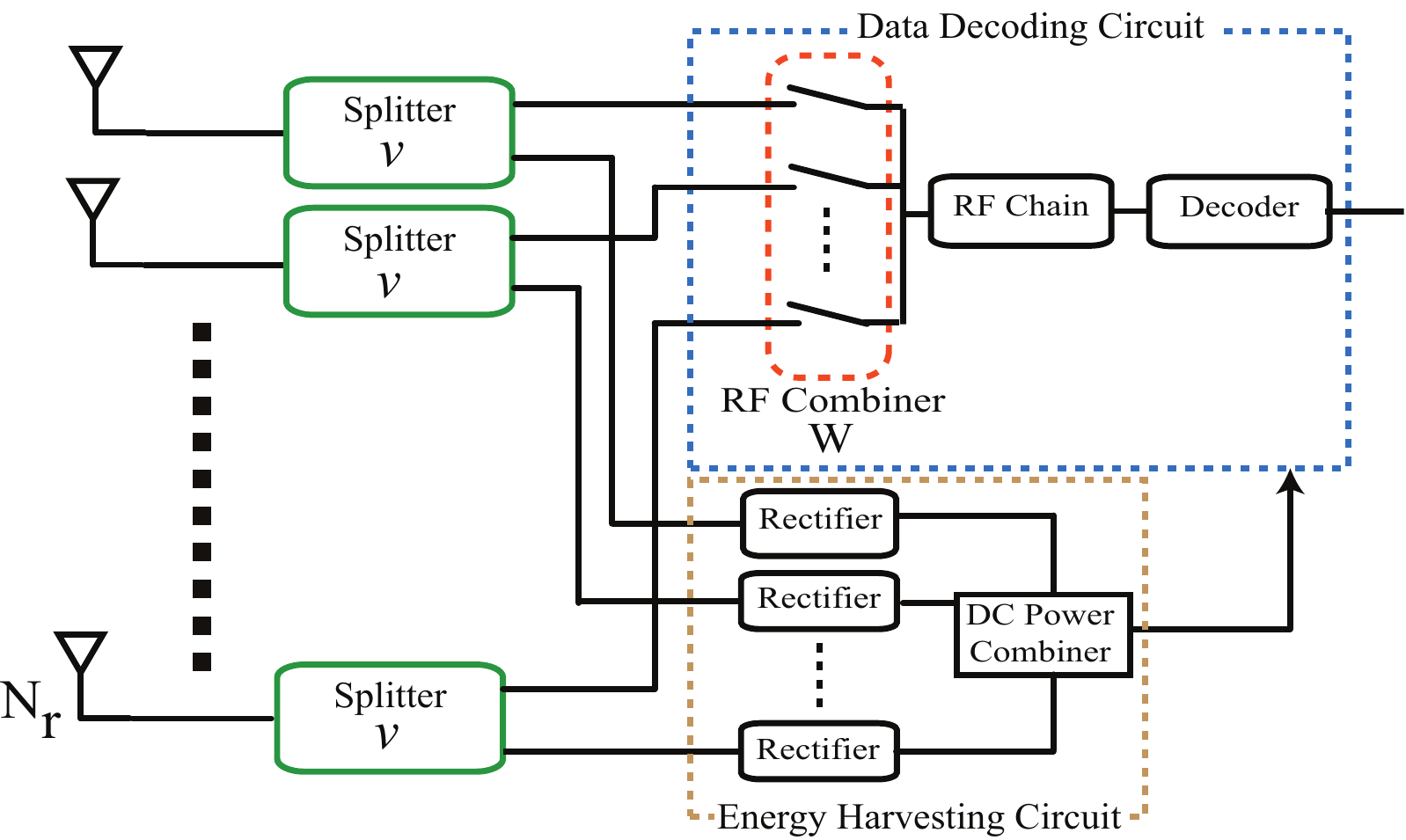}
	}
	\caption{Low power receiver architecture for SWIPT.}
	\label{fig:Model_Switches_v1}
\end{figure}

We now derive the combining gain expression for the proposed SWIPT receiver architecture in Fig. \ref{fig:Model_Switches_v1}. Let $y$ be the signal output at the RF combiner. If a BS applies a beamforming vector $\bff\in\mathbb{C}^{N_\textrm{t}\times 1}$ to send data symbol $s$ $\left(\text{where}\,\, \bbE\left[|s|^2\right]=P_\textrm{t}\right)$ to a target user, it follows that
\begin{equation}
y=\sqrt{\nu}\left[\bw^* \bH_\mathrm{d} \bff s + \bw^* \br_\mathrm{int} + \bw^* \bn\right],
\label{eq:Received_Signal}
\end{equation}
where $\bH_\mathrm{d}\in\mathbb{C}^{N_\textrm{r}\times N_\textrm{t}}$ is the channel between the user and its serving BS,  and $\br_\mathrm{int}$ is the received signal due to the interfering BSs. Since the channel between each user and its BS is assumed to be single-path, the channel matrix 
$\bH_\mathrm{d}= h_0\sqrt{g_0(r_0)} \ba_\mathrm{r}\left(\phi_\mathrm{r}\right) \ba_\mathrm{t}^{*} \left(\phi_\mathrm{t}\right)$,
%
where $\ba_\mathrm{t}\left(\phi_\mathrm{t}\right) $ and $\ba_\mathrm{r}\left(\phi_\mathrm{r}\right)$ are the array response vectors at the BS and user, respectively. Recall that $g_0(r_0)$ denotes the path gain from the serving BS, while $\phi_\mathrm{r}$ and $\phi_\mathrm{t}$ respectively denote the channel angle-of-arrival and angle-of-departure at the user and BS. If the channel is known at the BS, and given the antenna model in Section \ref{secAnt}, the BS will design the beamforming vector $\bff$ to maximize the beamforming gain, i.e., to have $\left|\ba_\mathrm{t}^*\left(\phi_\mathrm{t}\right) \bff \right|^2=N_\textrm{t}$. Denoting $\bar{\alpha}=h_0\sqrt{g_0(r_0)} \ba_\mathrm{t}^*\left(\phi_\mathrm{t}\right) \bff$, the received signal in \eqref{eq:Received_Signal} can be written as
\begin{equation}
y=\sqrt{\nu}\left( \bar{\alpha} \bw^* \ba_\mathrm{r}\left(\phi_\mathrm{r}\right) s + \bw^* \br_\mathrm{int} + \bw^* \bn \right).
\end{equation}
The post-combining SINR can then be expressed as
\begin{equation}\label{eq:maxsinr}
\text{SINR}=\frac{\nu P_\textrm{t}   \left|\alpha\right|^2 N_t \left|\bw^* \ba_\mathrm{r}\left(\phi_\mathrm{r}\right)\right|^2}{I+ \nu\bw^* \bw \sigma^2},
\end{equation}
where $\left| \bw^*  \ba_\mathrm{r}\left(\phi_\mathrm{r}\right) \right|^2$ represents the combining gain at the receiver, and $I$ denotes the aggregate interference power. The SINR in (\ref{eq:maxsinr}) can be maximized if the receiver designs the optimum combining vector, which can be implemented by activating certain antennas on or off.
This requires the receiver to have global channel knowledge, which is often challenging in practice. We relax this condition by assuming that the receiver has the angle-of-arrival information for the serving BS only. Ignoring the interference, we propose to design the combining vector by maximizing the $\text{SNR}=\frac{P_\textrm{t}\left|\bar{\alpha}\right|^2 N_t \left|\bw^* \ba_\mathrm{r}\left(\phi_\mathrm{\phi_\textrm{r}}\right)\right|^2} {\bw^* \bw \sigma^2}$ instead, i.e.,
the receiver designs its combining vector $\bw$ such that
\begin{align}
\bw^{\star}=\arg\max_{\bw\,\in \,[0,1]^{N_\textrm{r}}} \frac{\left|\bw^* \ba_\mathrm{r}\left(\phi_\mathrm{r}\right)\right|^2} {\bw^* \bw }.
\label{eq:Opt_Arch1}
\end{align}
The optimal solution to \eqref{eq:Opt_Arch1} can be found by an exhaustive search over all possible combinations of $\bw$. For large receive antenna arrays, this could entail high computational costs, which would further increase the power consumption. Therefore, it is important to consider computationally efficient approaches for designing the combining vector.
As outlined in Algorithm \ref{alg:alg_Arch1}, we propose a greedy solution for designing $\bw$ by (step-wise) activating only those antennas that boost the received SNR. 
With $\hat{\mathbf{w}}$ denoting the combining vector designed using Algorithm 1, the combining gain for the switch-based architecture can be defined as $M_c=\frac{\left|\sum_{i=1}^{N_r} \hat{\textrm{{w}}}_i e^{\j k d (i-1)  \cos\left(\phi_\mathrm{r}\right)} \right|^2}{|\hat{\mathbf{w}}|^2}$ where $k$ denotes the wavenumber and $d$ is the antenna element spacing.
Despite its low-complexity, numerical simulations in the next section reveal that our low-power greedy approach could give a good combining gain, without losing substantial performance compared to more advanced but power-hungry solutions.
\begin{algorithm} [!t]                     
	\caption{Greedy Switch Combining Design}          
	\label{alg:alg_Arch1}                           
	\begin{algorithmic} 
		\State \textbf{Input} $N_\textrm{r}$, $\phi_\mathrm{r}$
		\State \textbf{Initialization} $\bw=\boldsymbol{0}$, $\textrm{w}_1=1$
		\For {$i=2,\cdots,N_\mathrm{r}$}
		\If {$\frac{1}{i} \left|\sum_{n=1}^{i-1} \textrm{w}_n e^{\j k d (n-1)  \cos\left(\phi_\mathrm{\textrm{r}}\right)} + e^{\j k d (i-1) \cos\left(\phi_\mathrm{\textrm{r}}\right)}\right|^2 > \frac{1}{i-1}\left|\sum_{n=1}^{i-1} \textrm{w}_n e^{\j k d (n-1) \cos\left(\phi_\mathrm{\textrm{r}}\right)} \right|^2$}
		\State	$\textrm{w}_i=1$
		\EndIf
		\EndFor
	\end{algorithmic}
\end{algorithm}
\subsection{Results}
Fig. \ref{fig:Cov_3D} plots the overall success probability for a given transmit antenna beam pattern. The users are equipped with a single-antenna receiver, similar to the one in Fig. \ref{fig:Model_Switches_v1} with $N_\mathrm{r}=1$.
First, Fig. \ref{fig:Cov_3D} shows that a reasonable success probability can be obtained with mmWave SWIPT system for typical mmWave propagation and system parameters. Further, this plot illustrates that the power splitting ratio $\nu$ needs to be optimized for a given SINR outage threshold to maximize the overall success probability. Matching the intuition, the figure shows that in the low SINR outage regime (when $T$ is large), it is desirable to divert more power to the information decoding module, while a larger fraction of power needs to be portioned for the energy harvesting system in the high SINR outage regime (when $T$ is small). This trend is consistent with prior studies on SWIPT architectures\cite{WPCSurvey2015}. 

We now evaluate the performance of the proposed low-power receiver architecture for different number of receive antennas.
In Fig. \ref{fig:Cov_Switch}, the success probability $P_\textrm{suc}(\lambda,T,\psi,\nu)$ is plotted for a fixed transmit antenna beam pattern.
For the proposed architecture, the combining vector is obtained using Algorithm 1, and the curves are averaged over the angle-of-arrival parameter.
For comparison, we also plot the success probability for (fully digital) maximal ratio combining (MRC) receivers.
We observe that the success probability improves with the receive antenna array size. 
Further, when the SINR outage threshold $T$ is small, the success probability is mainly limited by the energy outage. This also explains why the success probability converges to a limit (determined by the energy outage threshold) as $T$ decreases. 
Moreover, there are diminishing returns as the number of antennas are increased. A comparison with power-hungry MRC receivers shows that the proposed switch-based architecture performs reasonably well. This is particularly desirable for future mmWave SWIPT devices.
\begin{figure} [t]
	\centerline{
		\includegraphics[width=\columnwidth]{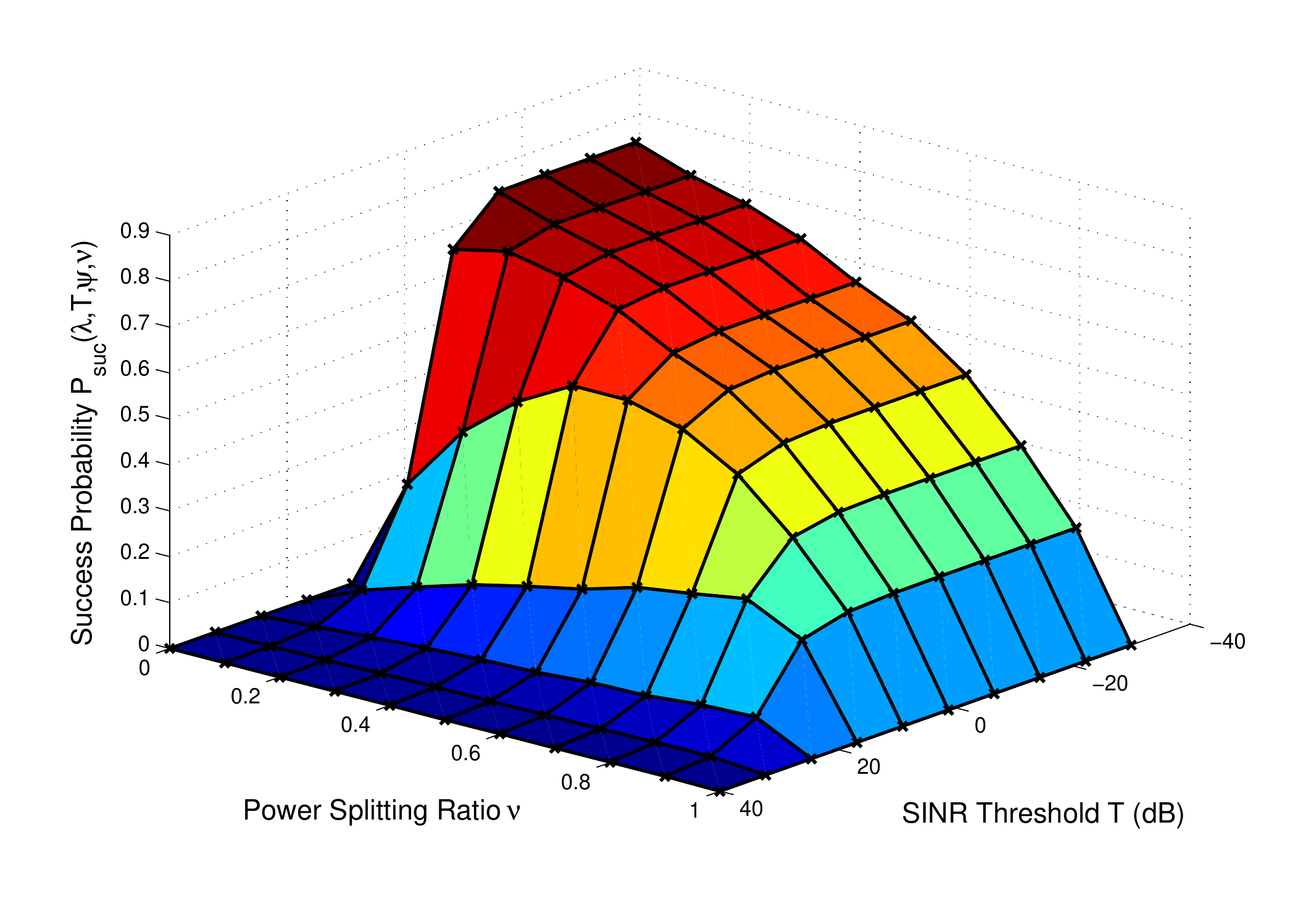}
	}
	\caption{A 3D plot showing the interplay between the success probability, the power splitting ratio $\nu$, and the SINR outage threshold $T$ for a given energy outage threshold $\psi$ and network density $\lambda$. As $T$ gets large, the system becomes SINR-limited, and the optimum value of $\nu$ increases, suggesting that a larger fraction of received signal should be used for information extraction to optimize the overall success probability. The transmit antenna beam pattern is set to $A_{15,-15,10^\circ, 350^\circ}$. Other parameters include $\psi=-70$ dB, $\sigma^2_c=-80$ dB, $\lambda=200/\text{km}^2$, and $P_\textrm{t}=43$ dBm.}
	\label{fig:Cov_3D}
\end{figure}

\begin{figure} [t]
	\centerline{
		\includegraphics[width=\columnwidth]{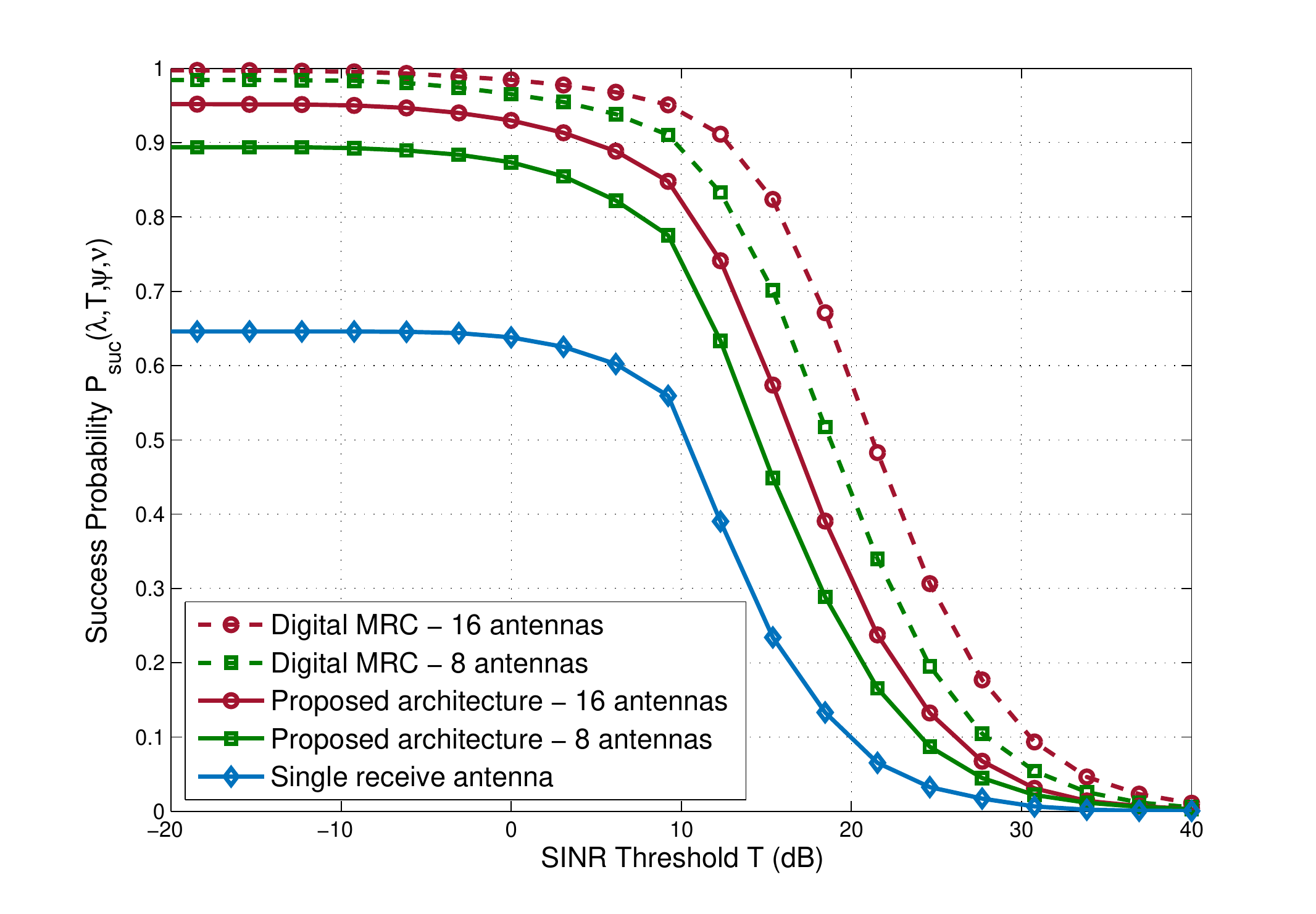}
	}
	\caption{The success probability for different number of receive antennas $N_{\textrm{r}}$ at the user given a fixed transmit beam pattern $A_{15,-15,10^\circ, 350^\circ}$ at the BSs. Proposed low-power architecture achieves good performance compared to superior receiver architectures. Other parameters include $\nu=0.5$, $\psi=-70$ {dB}, $\sigma^2_c=-80$ dB, $\lambda=200/\text{km}^2$, and $P_\textrm{t}=43$ dBm.}
	\label{fig:Cov_Switch}
\end{figure}
\section{Conclusions}\label{secConc}
In this paper, we analyzed the energy harvesting performance at low-power devices powered by a mmWave cellular network. Using a stochastic geometry framework, we derived analytical expressions characterizing the performance of mmWave energy and information transfer in terms of system, channel and network parameters. Simulations results were used to validate the accuracy of the derived expressions.
Leveraging the analytical framework, we also provided useful network and device level design insights. For the connected case when the transmitter and receiver beams are aligned, results show that the energy coverage improves with narrower beams. In contrast, wider beams provide better energy coverage when the receivers are not aligned with a particular transmitter. This trade-off is evident in the more general scenario having both types of receivers, where there typically exists an optimal beamforming beamwidth that maximizes the network-wide energy coverage. 
Moreover, we found that several device-related parameters can significantly impact the system performance. For example, the performance can be substantially improved by optimizing over the power splitting ratio and by leveraging large antenna arrays. 
To allow using multiple antennas
at the mmWave receivers while keeping the power consumption low, we proposed a low-power receiver architecture for mmWave energy and information transfer
using antenna switches. Simulation results show that the proposed architecture
can provide good gains for the overall mmWave energy harvesting performance. Simulation results also reveal that mmWave cellular networks could potentially provide better energy coverage than lower frequency solutions.

\appendices
\allowdisplaybreaks
\section*{Appendix A: Theorem 1}
The following inequality approximates the tail probability of a normalized Gamma distribution.
\begin{lemma}[From \cite{alzer1997some}]\normalfont
For a normalized Gamma random variable $u$ with parameter $N$, the probability $\Pr\left(u<x\right)$ can be tightly upper-bounded by
$\Pr\left(u<x\right) < \left(1-e^{-ax}\right)^N$,
where the constant $x>0$ and $a=N(N!)^{-\frac{1}{N}}$.
\end{lemma}
We write $\textrm{P}_\textrm{con}\left(\lambda,\psi\right)=\Pr\left[Y>\max\left(\frac{\psi}{\xi},\psi_\textrm{min}\right)\right]=\Pr\left[S+I>\hat{\psi}\right]$, where $S=P_\textrm{t}M_\textrm{t}M_\textrm{r}H_0g_0(r_0)$ is the received signal power from the serving BS,
and $I=\sum_{\ell>0,\ell\in\Phi(\lambda)\setminus\mathbb{B}(r_g)}^{} P_\textrm{t}\delta_\ell H_\ell g_\ell(r_\ell)$ is the received signal power from all the other BSs. We can derive the result in Theorem \ref{thm1} by finding the conditional distributions $\textrm{P}_\textrm{con,L}\left(\lambda,\psi\right)$ and $\textrm{P}_\textrm{con,N}\left(\lambda,\psi\right)$. To proceed, first consider the conditional distribution $\textrm{P}_\textrm{con,L}\left(\lambda,\psi\right)=\Pr\left(S+I>\psi|\los\right)$ given the receiver is aligned with a LOS BS (which is indicated by the subscript $\los$ in the following notation).
\begin{small}
\begin{align}\label{temp}
\textrm{P}_\textrm{con,L}\left(\lambda,\psi\right)&\overset{}{=}\mathbb{E}_{S,I|\los}\left[\Pr\left(u<\frac{S+I}{\psi}\right)\right] \nonumber\\ &\overset{(a)}{\approx}\mathbb{E}_{S,I|\los}\left[\left(1-e^{-a\frac{S+I}{\psi}}\right)^{{N}}\right]  \nonumber \\
&\overset{}{=}\mathbb{E}_{S,I|\los}\left[\sum\limits_{k=0}^{N}\left(-1\right)^{k}\binom{N}{k}e^{-a k\frac{S+I}{\psi}}\right] \nonumber\\ &\overset{}{=}\sum\limits_{k=0}^{N}\left(-1\right)^{k}\binom{N}{k}\mathbb{E}_{S,I|\los}\left[e^{-\hat{a} \left(S+I\right)}\right]
\end{align}\end{small}where we have included a dummy random variable $u\, ${\raise.17ex\hbox{$\scriptstyle\mathtt{\sim}$}}$ \,\Gamma\left(N,\frac{1}{N}\right)$ in the first equation. Note that $u$ converges to $1$ as $N\to\infty$. Therefore, this substitution is in fact an approximation when $N$ is finite. The introduction of $u$ allows leveraging the inequality in Lemma 5, which leads to $(a)$, where the constant $a=N(N!)^{-\frac{1}{N}}$. The last equation follows from the Binomial series expansion of $(b)$, and by further substituting $\hat{a}=\frac{ak}{\psi}$. To evaluate the expectation in (\ref{temp}), consider
\begin{align}\label{temp2}
&\mathbb{E}_{S,I|\los}\left[e^{-\hat{a} \left(S+I\right)}\right]=\mathbb{E}_{S|\los}\left[e^{-\hat{a}S} \mathbb{E}_{I|{S,\los}}\left[e^{-\hat{a}I}\right]\right]. 
\end{align}
The inner expectation in (\ref{temp2}) can be simplified by applying the thinning theorem for a PPP \cite{haenggi2012stochastic}. Note that $\Phi$ can be independently thinned into two PPPs $\Phi_\los$ and $\Phi_\nlos$, where the former comprises the LOS BSs whereas the latter consists of NLOS BSs.
Therefore, we can interpret $\Phi_\mathrm{L}$ and $\Phi_\mathrm{N}$ as two independent tiers of BSs. The user will be tagged with either the closest BS in $\Phi_\mathrm{L}$ or in $\Phi_\mathrm{N}$, whichever maximizes the average received power at the user.
We can further thin $\Phi_\los$ into four independent PPPs $\{\Phi_\los^i\}_{i=1}^{4}$, where each resulting PPP $\Phi_\los^i$ contains BSs that correspond to a nonzero directivity gain $D_i$ with $p_i$ being the thinning probability.
This follows because the beam orientations are assumed to be independent across links. Thus, a link can have a directivity gain of $D_i$ with probability $p_i$ independently of other links.
We let the received power due to the transmission from the BSs in $\Phi_\los^i$ be $I^i_\los$.
Likewise, $\Phi_\nlos$ can be split into $\{\Phi_\nlos^i\}_{i=1}^{4}$ with the corresponding received powers denoted by $\{I^i_\nlos\}_{i=1}^{4}$. Since the resulting PPPs are independent, (\ref{temp2}) can be simplified as
\begin{small}
\begin{align}\label{temp4}
\mathbb{E}_{I|S,\los}\left[e^{-\hat{a}I}\right]=\prod\limits_{i=1}^{4} \mathbb{E}_{I|{S,\los}}\left[e^{-\hat{a}I^{i}_\los}\right]
\prod\limits_{j=1}^{4}\mathbb{E}_{I|{S,\los}}\left[e^{-\hat{a}I^{j}_{\nlos}}\right]
\end{align}
\end{small}where 
\begin{small}
\begin{align}\label{temp3}
&\mathbb{E}_{I|S,\los}\left[e^{-\hat{a}I^{i}_{\los}}\right]
\overset{(a)}{=}\mathbb{E}_{\Phi_{\mathrm{L}}^i|r_o}\left[\prod\limits_{\ell\in\Phi_{\mathrm{L}}^i\setminus\mathbb{B}\left(r_o\right)} \mathbb{E}_{H_{\ell}}\left[e^{-\hat{a}P_\textrm{t}H_{\ell}D_iC_{\mathrm{L}}{r_{\ell}}^{-{\alpha}_{\mathrm{L}}}}\right]\right] \nonumber\\
&\overset{(b)}{=}\mathbb{E}_{\Phi_{\mathrm{L}}^i|r_o}\left[\prod\limits_{\ell\in\Phi_{\mathrm{L}}^i\setminus\mathbb{B}\left(r_o\right)} \left(\frac{1}{1+\hat{a}P_\textrm{t}D_iC_{\mathrm{L}} {r_{\ell}}^{-{\alpha}_{\mathrm{L}}}{N_\mathrm{L}}^{-1} } \right)^{N_\mathrm{L}}\right]\nonumber\\
&=e^{-2\pi\lambda p_i\int\limits_{r_o}^{\infty}\left(1-\left(\frac{1}{1+\hat{a}P_\textrm{t}D_iC_{\mathrm{L}} {t}^{-{\alpha}_{\mathrm{L}}}{N_\mathrm{L}}^{-1} }\right)^{N_{\mathrm{L}}}\right)p(t)t\text{d}t }
\end{align}\end{small}where $(a)$ follows by conditioning on the length $r_o$ of the serving LOS link, and by further noting that small-scale fading is independent across links. Here, $\mathbb{B}\left(r_o\right)$ denotes a circular disc of radius $r_o$ centered at the typical user. $(b)$ is obtained by using the moment generating function of a normalized Gamma random variable, while the last equation follows by invoking the probability generating functional\cite{haenggi2012stochastic} of the PPP $\Phi_L^{i}$. 
Substituting (\ref{temp3}) in the first (left) product term of (\ref{temp4}) yields (\ref{fac1}).
Similarly, $\mathbb{E}_{I|S,\los}\left[e^{-\hat{a}I^{i}_{\mathrm{N}}}\right]$ is given by
\begin{small}
\begin{align}\label{temp5}
&\mathbb{E}_{\Phi_\los^i,H|r_o}\left[e^{-\hat{a}\sum\limits_{\ell\in\Phi_{\mathrm{N}}^i\setminus \mathbb{B}\left(\rho_{\mathrm{L}}(r_o)\right)}P_\textrm{t}H_{\ell}D_iC_\mathrm{N}{r_\ell}^{-\alpha_{\mathrm{N}}}}\right]\nonumber\\
&=e^{-2\pi\lambda p_i\int\limits_{\rho_{\mathrm{L}}(r_o)}^{\infty}\left(1-\left(\frac{1}{1+\hat{a}P_\textrm{t}D_iC_{\mathrm{N}} {x}^{-{\alpha}_{\mathrm{N}}}{N_\mathrm{N}}^{-1} }\right)^{N_{\mathrm{N}}}\right)\left(1-p(t)\right)t\text{d}t }.
\end{align}
\end{small}By substituting (\ref{temp5}) in the second (right) product term of (\ref{temp4}) yields (\ref{fac2}). Using the expressions in (\ref{temp4})--(\ref{temp5}) in (\ref{temp2}), and by further evaluating the expectation of the resulting expression with respect to $S$, we obtain
\begin{small}
\begin{align}
&\int\limits_{r_g}^{\infty}\left({\frac{1}{1+\hat{a}P_\textrm{t}M_{\textrm{t}}M_{\textrm{r}}C_{\textrm{L}}{{r}^{-{\alpha}_{\mathrm{L}}}}{ {N_{\mathrm{L}}}^{-1} } } } \right)^{N_{\mathrm{L}}} 
\nonumber\\
&\quad\qquad\times e^{-{\overset{}{\Upsilon}}_{k,1}\left(\lambda,\psi,r\right) -{\overset{}{\Upsilon}_{k,2}\left(\lambda,\psi,\rho_\los(r)\right)}}  {\tilde{\tau}}_{\mathrm{L}}(r) \text{d}r
\end{align}
\end{small}where we have again used definition of the moment generating function of a normalized Gamma distribution. ${\overset{}{\Upsilon}}_{k,1}\left(\cdot\right)$ and ${\overset{}{\Upsilon}}_{k,2}\left(\cdot\right)$ are given in (\ref{fac1}) and (\ref{fac2}) respectively, $r_g$ denotes the minimum link distance, while the distance distribution is provided in Lemma 3. Using (\ref{temp}) and (\ref{temp2}), we can thus retrieve the expression in (\ref{los}).
We can similarly derive the conditional distribution $\textrm{P}_\textrm{con,N}\left(\lambda,\psi\right)=\Pr\left(S+I>\psi|\nlos\right)$ in (\ref{nlos}) for the NLOS case.
\balance
\section*{Appendix B: Proposition 1}
To simplify the analysis, we approximate the LOS probability function $p(r)$ by a step function $p(r)=\mathbbm{1}_{\{0<r<R_B\}}$, i.e., the BSs within a LOS ball of radius $R_B$ are marked LOS with probability 1, while the rest as NLOS\cite{bai2015}.     
The radius $R_B$ is chosen such that the LOS association probability $\varrho_\mathrm{L}$ remains the same (as in our original model). Using a step function for $p(r)$, it follows from Lemma 2 that $\varrho_\mathrm{L}=1-e^{-\lambda\pi R_B^2}$ and $R_B=\left(\frac{\ln\left(1-\varrho_\mathrm{L}\right)}{\lambda\pi}\right)^{0.5}$. Moreover, we neglect small scale fading for all the links except for the serving BS. We also ignore the NLOS signals in the analysis. This effectively leads to a scenario where the user receives signals from the BSs within the LOS ball only. Intuitively, this would be the likely scenario in sufficiently dense networks. 
We only list the key steps since the rest of the proof follows from Appendix A. Ignoring the NLOS signals, we approximate (\ref{temp4}) as $\mathbb{E}_{I|S,\los}\left[e^{-\hat{a}I}\right]\approx\prod\limits_{i=1}^{4} \mathbb{E}_{I|{S,\los}}\left[e^{-\hat{a}I^{i}_\los}\right]$ where
\begin{small}
\begin{align}\label{emp3}
\mathbb{E}_{I|S,\los}\left[e^{-\hat{a}I^{i}_{\los}}\right]&\overset{}{=}\mathbb{E}_{\Phi_{\mathrm{L}}^i|r_o}\left[\prod\limits_{\ell\in\Phi_{\mathrm{L}}^i\cap[\mathbb{B}\left(R_B\right)\setminus\mathbb{B}\left(r_o\right)]} e^{-\hat{a}P_\textrm{t}D_iC_{\mathrm{L}}{r_{\ell}}^{-{\alpha}_{\mathrm{L}}}}\right]\nonumber\\
&\overset{(a)}{=}e^{-2\pi\lambda p_i\int\limits_{r_o}^{R_B}\left(1-
e^{-\hat{a}P_\textrm{t}D_iC_{\mathrm{L}}{t}^{-{\alpha}_{\mathrm{L}}}}
\right)t\text{d}t } \nonumber\\
&\overset{(b)}{=}e^{-\frac{2\pi\lambda p_i}{\alpha_\los}\int_{W_{ik}r_o^{-\alpha_\los}}^{W_{ik}R_B^{-\alpha_\los}}\frac{1-e^{-v}}{v^{1+\frac{2}{\alpha_\los}}}\text{d}v }\nonumber\\
&\overset{(c)}{=}
e^{-\frac{2\pi\lambda p_iW_{ik}^{\frac{2}{\alpha_\los}}}{\alpha_\los}\Gamma\left(\frac{-2}{\alpha_\los}; W_{ik}r_o^{-\alpha_\los}, W_{ik}R_B^{-\alpha_\los}\right)}\nonumber\\
&\hspace{0.5in}\times e^{-{\pi\lambda p_i\left(R_B^2-r_o^2\right)}}.
\end{align}
\end{small}Here, (a) follows by ignoring the small scale fading and invoking the probability generating functional\cite{haenggi2012stochastic} of PPP, (b) by a change of variables, and (c) by the definition of the generalized incomplete Gamma function. The result in Proposition 1 is obtained by assuming $p_5=0$, and by further noting that the distance distribution simplifies to ${\tau_\los(x)}=\frac{2\pi\lambda x}{\varrho_\los}e^{-\lambda\pi x^2}$ due to the LOS ball approximation.
\section*{Appendix C: Corollary 1}
We derive Corollary 1 by finding the conditional means $\bar{P}_\los=\mathbb{E}\left[S+I|\los\right]$ and $\bar{P}_\nlos=\mathbb{E}\left[S+I|\nlos\right]$.
$\bar{P}_\los$ can be evaluated by conditioning on the link distance $r_0$ from the serving BS as follows. 
\begin{small}
\begin{align}\label{appb}
&\mathbb{E}_{}\left[S+I|r_0,\los\right]=
\mathbb{E}_{}\left[S|r_0,\los\right] +\sum_{i=1}^{4}\mathbb{E}_{}\left[I_\los^i+I_\nlos^i|r_0,\los\right] \nonumber\\
&\overset{(a)}{=}P_\textrm{t} M_{\textrm{t}}M_{\textrm{r}}C_{\mathrm{L}} {r_0}^{-{\alpha}_{\mathrm{L}}}+\sum_{i=1}^{4}2\pi\lambda P_\textrm{t}p_iD_i C_{\mathrm{L}}\int\limits_{r_0}^{\infty}t^{-\left({\alpha}_{\mathrm{L}}-1\right)}p(t)\text{d}t\nonumber\\
&+\sum_{i=1}^{4}2\pi\lambda P_\textrm{t} p_i D_i C_\nlos\left[\displaystyle\frac{\left(\rho_\los(r_0)\right)^{-(\alpha_\nlos-2)}}{\alpha_\nlos-2}
-\int\limits_{{\rho}_\los(r_0)}^{\infty} t^{-\left({\alpha}_\nlos-1\right)}p(t)\text{d}t\right] \nonumber\\
&=P_\textrm{t} M_{\mathrm{t}} M_{\mathrm{r}} C_\los {r_0}^{-\alpha_\los}+{\Psi}_\los\left(r_0\right)+{\Psi}_\nlos\left(\rho_\los(r_0)\right)
\end{align}
\end{small}where $(a)$ is obtained by averaging over the fading distribution, followed by invoking Campbell's theorem\cite{haenggi2012stochastic}, while (\ref{appb}) follows from the definitions of ${\Psi}_\los$ and ${\Psi}_\nlos$ provided in (\ref{fac3}) and (\ref{fac4}) respectively. Taking expectation of $\mathbb{E}_{}\left[S+I|r_0,\los\right]$ with respect to $r_o$ using Lemma 3 yields (\ref{alos}). The expression for $\bar{P}_\nlos$ is (\ref{anlos}) can be derived using similar steps.
\section*{Appendix D: Theorem 3}
From (\ref{eq1}), it follows that the harvested energy $\gamma=\xi Y\mathbbm{1}_{\{Y>\psi_\textrm{min}\}}$. Let $Y=(1-\nu)\left(S+I+\sigma^2\right)$, where
$S=
P_\textrm{t}M_\textrm{t}M_\textrm{r} H_0 g_0(r_0)$ and $I=\sum_{\ell>0,\ell\in\Phi(\lambda)\setminus\mathbb{B}(r_g)}^{}
P_\textrm{t}\delta_\ell H_\ell g_\ell(r_\ell)$ respectively denote the contributions from the serving and the interfering BSs. To find $P_\text{suc}\left(\lambda,T,\psi,\nu\right)=\Pr\left[\text{SINR}>T,\gamma>{\psi}\right]$, consider
\begin{small}
\begin{align}\label{eq:finalapp}
\Pr&\left[\frac{\nu S}{\nu(I+\sigma^2)+\sigma_c^2}>T, (1-\nu)\left(S+I+\sigma^2\right)>\hat{\psi}\right] \nonumber \\
&\overset{(a)}{=}{\mathbb{E}}_{I}\left[\Pr\left[S>T\left(I+\sigma^2+\frac{\sigma^2_c}{\nu}\right), S> \frac{\hat{\psi}}{\left(1-\nu\right)}-I-\sigma^2\right]\right] \nonumber \\
&\overset{(b)}{=}{\mathbb{E}}_{I}\left[\Pr\left[S>T\left(I+\sigma^2+\frac{\sigma^2_c}{\nu}\right)\right]\bigg|I>\mu\right]\Pr\left[I>\mu\right]\nonumber\\
&\quad\qquad+{\mathbb{E}}_{I}\left[\Pr\left[S> \frac{\hat{\psi}}{\left(1-\nu\right)}-I-\sigma^2\right]\bigg|I\leq\mu\right]
\Pr\left[I\leq\mu\right] \nonumber \\
&\overset{(c)}{\approx}{\mathbb{E}}_{I}\left[\Pr\left[S>T\left(I+\sigma^2+\frac{\sigma^2_c}{\nu}\right)\right]\right]\Pr\left[I>\mu\right]\nonumber\\
&\qquad\qquad\qquad+{\mathbb{E}}_{I}\left[\Pr\left[S> \frac{\hat{\psi}}{\left(1-\nu\right)}-I-\sigma^2\right]\right]
\Pr\left[I\leq\mu\right] \nonumber \\
&\overset{}{=}P_{\text{cov}}(\lambda,T,\nu)\tilde{P}_{\textrm{con}}(\lambda,\mu)
+{P}_{\text{con}}(\lambda,\varphi)\left[1-\tilde{P}_{\textrm{con}}(\lambda,\mu)\right]
\end{align}
\end{small}where the expectation in ($a$) is with respect to the interference $I$, ($b$) is obtained by further conditioning on $I$ to be greater (or smaller) than  a parameter $\mu$ 
which follows from the inequality
$T\left(I+\sigma^2+\frac{\sigma^2_c}{\nu}\right)>\frac{\hat{\psi}}{\left(1-\nu\right)}-I-\sigma^2$. The approximation (or effectively an upperbound) in ($c$) results from dropping the conditions $I>\mu$ or $I\leq\mu$ while calculating the expectation. Finally, the SINR coverage probability $P_{\textrm{cov}}(\lambda,T,\nu)$ follows from Lemma \ref{lem4}, and the energy coverage probability $P_{\text{con}}(\lambda,\varphi)$ from Theorem \ref{thm1}.
$\tilde{P}_{\textrm{con}}(\lambda,\mu)$ is the interference CCDF evaluated at $\mu$.

\bibliographystyle{ieeetr}

\begin{thebibliography}{10}

\bibitem{talha2015gc}
T.~A. Khan {\em et~al.}, ``Energy coverage in millimeter wave energy harvesting
  networks,'' in {\em 2015 IEEE Globecom Workshops (GC Wkshps)}, pp.~1--6, Dec.
  2015.

\bibitem{rappaport2014millimeter}
T.~S. Rappaport {\em et~al.}, {\em Millimeter Wave Wireless Communications}.
\newblock Pearson Education, 2014.

\bibitem{bai2014b}
T.~Bai {\em et~al.}, ``Coverage and capacity of millimeter-wave cellular
  networks,'' {\em IEEE Commun. Mag.}, vol.~52, pp.~70--77, Sept. 2014.

\bibitem{EnergyHarvestWirelessCommSurvey2015}
S.~Ulukus {\em et~al.}, ``Energy harvesting wireless communications: A review
  of recent advances,'' {\em IEEE J. Sel. Areas Commun.}, vol.~33,
  pp.~360--381, Mar. 2015.

\bibitem{IoT2014}
A.~Zanella {\em et~al.}, ``Internet of things for smart cities,'' {\em IEEE
  Internet Things J.}, vol.~1, pp.~22--32, Feb. 2014.

\bibitem{rangan2014millimeter}
S.~Rangan {\em et~al.}, ``Millimeter-wave cellular wireless networks:
  Potentials and challenges,'' {\em Proc. IEEE}, vol.~102, pp.~366--385, Mar.
  2014.

\bibitem{talla2015powering}
V.~Talla {\em et~al.}, ``Powering the next billion devices with {W}i-{F}i,''
  {\em ar{X}iv preprint ar{X}iv:1505.06815}, 2015.

\bibitem{GollakotaRF}
S.~Gollakota {\em et~al.}, ``The emergence of {RF}-powered computing,'' {\em
  Computer}, vol.~47, pp.~32--39, Jan. 2014.

\bibitem{RFsurveyLondon}
M.~Pinuela {\em et~al.}, ``Ambient {RF} energy harvesting in urban and
  semi-urban environments,'' {\em IEEE Trans. Microw. Theory Techn.}, vol.~61,
  pp.~2715--2726, Jul. 2013.

\bibitem{valenta2014}
C.~Valenta and G.~Durgin, ``Harvesting wireless power: Survey of
  energy-harvester conversion efficiency in far-field, wireless power transfer
  systems,'' {\em IEEE Microw. Mag.}, vol.~15, pp.~108--120, Jun. 2014.

\bibitem{WPCSurvey2015}
S.~Bi {\em et~al.}, ``Wireless powered communication: opportunities and
  challenges,'' {\em IEEE Commun. Mag.}, vol.~53, pp.~117--125, Apr. 2015.

\bibitem{zhang2013mimo}
R.~Zhang and C.~K. Ho, ``{MIMO} broadcasting for simultaneous wireless
  information and power transfer,'' {\em IEEE Trans. Wireless Commun.},
  vol.~12, pp.~1989--2001, May 2013.

\bibitem{SWIPTfading}
L.~Liu {\em et~al.}, ``Wireless information transfer with opportunistic energy
  harvesting,'' {\em IEEE Trans. Wireless Commun.}, vol.~12, pp.~288--300, Jan.
  2013.

\bibitem{SWIPTInter}
J.~Park and B.~Clerckx, ``Joint wireless information and energy transfer in a
  two-user {MIMO} interference channel,'' {\em IEEE Trans. Wireless Commun.},
  vol.~12, pp.~4210--4221, Aug. 2013.

\bibitem{flint2015}
I.~Flint {\em et~al.}, ``Performance analysis of ambient {RF} energy harvesting
  with repulsive point process modeling,'' {\em IEEE Trans. Wireless Commun.},
  vol.~PP, no.~99, pp.~1--1, 2015.

\bibitem{huang2013cog}
S.~Lee {\em et~al.}, ``Opportunistic wireless energy harvesting in cognitive
  radio networks,'' {\em IEEE Trans. Wireless Commun.}, vol.~12,
  pp.~4788--4799, Sep. 2013.

\bibitem{ekram2015}
A.~Sakr and E.~Hossain, ``Cognitive and energy harvesting-based {D2D}
  communication in cellular networks: Stochastic geometry modeling and
  analysis,'' {\em IEEE Trans. Commun.}, vol.~63, pp.~1867--1880, May 2015.

\bibitem{kaibin2014cellular}
K.~Huang and V.~Lau, ``Enabling wireless power transfer in cellular networks:
  Architecture, modeling and deployment,'' {\em IEEE Trans. Wireless Commun.},
  vol.~13, pp.~902--912, Feb. 2014.

\bibitem{kaibinlarsson2013}
K.~Huang and E.~Larsson, ``Simultaneous information and power transfer for
  broadband wireless systems,'' {\em IEEE Trans. Sig. Proc.}, vol.~61,
  pp.~5972--5986, Dec. 2013.

\bibitem{krikidis2014swipt}
I.~Krikidis, ``Simultaneous information and energy transfer in large-scale
  networks with/without relaying,'' {\em IEEE Trans. Commun.}, vol.~62,
  pp.~900--912, Mar. 2014.

\bibitem{bai2015}
T.~Bai and R.~Heath, ``Coverage and rate analysis for millimeter-wave cellular
  networks,'' {\em IEEE Trans. Wireless Commun.}, vol.~14, pp.~1100--1114, Feb.
  2015.

\bibitem{Singh2015}
S.~Singh {\em et~al.}, ``Tractable model for rate in self-backhauled millimeter
  wave cellular networks,'' {\em IEEE J. Sel. Areas Commun.}, vol.~33,
  pp.~2196--2211, Oct. 2015.

\bibitem{bai2014}
T.~Bai {\em et~al.}, ``Analysis of blockage effects on urban cellular
  networks,'' {\em IEEE Trans. Wireless Commun.}, vol.~13, pp.~5070--5083, Sep.
  2014.

\bibitem{hunter2008}
A.~Hunter {\em et~al.}, ``Transmission capacity of ad hoc networks with spatial
  diversity,'' {\em IEEE Trans. Wireless Commun.}, vol.~7, pp.~5058--5071, Dec.
  2008.

\bibitem{haenggi2012stochastic}
M.~Haenggi, {\em Stochastic geometry for wireless networks}.
\newblock Cambridge University Press, 2012.

\bibitem{tabesh2015power}
M.~Tabesh {\em et~al.}, ``A power-harvesting pad-less millimeter-sized radio,''
  {\em IEEE J. Solid-State Circuits}, vol.~50, pp.~962--977, Apr. 2015.

\bibitem{Charthad2016mmWaveWPT}
J.~Charthad, N.~Dolatsha, A.~Rekhi, and A.~Arbabian, ``System-level analysis of
  far-field radio frequency power delivery for mm-sized sensor nodes,'' {\em
  IEEE Trans. Circuits Syst. I: Reg. Papers}, vol.~63, pp.~300--311, Feb. 2016.

\bibitem{wu2015human}
T.~Wu, T.~S. Rappaport, and C.~M. Collins, ``The human body and millimeter-wave
  wireless communication systems: Interactions and implications,'' in {\em 2015
  IEEE Int. Conf. Commun. (ICC)}, pp.~2423--2429, Jun. 2015.

\bibitem{ghosh2010fundamentals}
A.~Ghosh {\em et~al.}, {\em Fundamentals of {LTE}}.
\newblock Pearson Education, 2010.

\bibitem{van2004detection}
H.~L. Van~Trees, {\em Detection, estimation, and modulation theory, optimum
  array processing}.
\newblock John Wiley \& Sons, 2004.

\bibitem{Mendez-Rial2015}
R.~M{\'e}ndez-Rial {\em et~al.}, ``Channel estimation and hybrid combining for
  mm{W}ave: Phase shifters or switches?,'' in {\em in Proc. of the Inf. Theory
  and Applications Workshop (ITA)}, 2015.

\bibitem{alzer1997some}
H.~Alzer, ``On some inequalities for the incomplete gamma function,'' {\em
  Mathematics of Computation of the American Mathematical Society}, vol.~66,
  no.~218, pp.~771--778, 1997.

\end{thebibliography}

\end{document}